\documentclass[11pt]{article}

\usepackage{authblk}
\usepackage{amssymb,latexsym,amsmath,amsthm,enumitem}
\usepackage{stmaryrd}
\usepackage{amsfonts}
\usepackage{fullpage}
\usepackage[latin1]{inputenc}

\theoremstyle{definition}
\newtheorem{theorem}{Theorem}
\newtheorem{corollary}[theorem]{Corollary}
\newtheorem{proposition}[theorem]{Proposition}

\newtheorem{lemma}[theorem]{Lemma}
\newtheorem{definition}[theorem]{Definition}
\newtheorem{example}[theorem]{Example}

\newtheorem{examples}[theorem]{Examples}
\newcommand\qbin[3]{\left[\begin{matrix} #1 \\ #2 \end{matrix} \right]_{#3}}

\newtheorem{notation}[theorem]{Notation}
\newtheorem{remark}[theorem]{Remark}

\usepackage{multirow}
\usepackage{graphicx}
\usepackage{tikz}
\usetikzlibrary{matrix,decorations.pathreplacing}
\usepackage{multirow}

\usepackage[all,cmtip]{xy}
\usepackage{hyperref}

\hypersetup{
  colorlinks   = true, 
  urlcolor     = blue, 
  linkcolor    = blue, 
  citecolor   = red 
}

\newcommand{\numberset}{\mathbb}
\newcommand{\N}{\numberset{N}}
\newcommand{\Z}{\numberset{Z}}

\newcommand{\R}{\numberset{R}}
\newcommand{\C}{\numberset{C}}

\newcommand{\F}{\numberset{F}}

\newcommand{\mS}{\mathcal{S}}
\newcommand{\mC}{\mathcal{C}}
\newcommand{\OO}{\mathcal{O}}

\newcommand{\mA}{\mathcal{A}}

\newcommand{\mB}{\mathcal{B}}

\newcommand{\mM}{\mathcal{M}}
\newcommand{\mD}{\mathcal{D}}

\newlength{\dhatheight}

\usepackage{accents}
\usepackage{enumitem}
\usepackage[margin=2.5cm]{geometry}
\usepackage{algpseudocode}
\usepackage{wrapfig}

\DeclareMathOperator{\ann}{ann}
\DeclareMathOperator{\rref}{RREF}
\DeclareMathOperator{\row}{row}
\DeclareMathOperator{\im}{Im}

\usepackage{relsize}
\usepackage{titling}

\begin{document}

\title{\textbf{An algebraic framework for end-to-end \\ physical-layer network coding}}

%
%
%


\author{Elisa Gorla and Alberto Ravagnani\thanks{E-mail addresses:
\texttt{elisa.gorla@unine.ch}, \texttt{ravagnani@ece.utoronto.ca}.
The authors were partially supported by the Swiss National Science
Foundation through grant no. 200021\_150207.}}
  
\affil{Institut de Math\'{e}matiques \\ Universit\'{e} de
Neuch\^{a}tel \\ Emile-Argand 11, CH-2000 Neuch\^{a}tel, Switzerland}

\date{}

\maketitle

\begin{abstract}
We propose an algebraic setup for end-to-end physical-layer network coding
based on submodule transmission. We introduce a distance function between modules,
describe how it relates to information loss and errors, and show how to compute it. 
Then we propose a definition
of submodule error-correcting code, and investigate bounds and constructions for such codes.
\end{abstract}

\section*{Introduction and motivation} 

In the framework of physical-layer network coding (PNC) multiple terminals attempt to exchange messages through intermediate relays.
The relays collect data from the terminals, and try to decode a function of the transmitted messages. 
Such function is then broadcasted to the terminals, which combine it with their side information to recover the other messages.

In \cite{nazer} the authors proposed a novel approach to PNC based on nested lattices, known as ``compute-and-forward''. 
Under this approach, the structure of a fixed underlying lattice is exploited by the relays to decode the function of the messages, which is then forwarded to the terminals. 
As observed in \cite{FSK}, this communication scheme induces an end-to-end network coding channel with channel equation 
\begin{equation}\label{cheq}
 Y=AX+Z.
\end{equation}
Here $X$ is the transmitted matrix, whose rows are elements from a given ambient space $\Omega$, $A$ is a transfer matrix, and $Z$ is an error matrix. 
In practice, $A$ and $Z$ are random matrices drawn according to certain distributions, that depend on the application at hand.

A general algebraic framework to study and construct nested-lattice-based PNC schemes was recently proposed in \cite{FSK} and further developed in~\cite{FNKS}. 
Following this algebraic description, which is compatible with any underlying lattice, the message space $\Omega$ has the structure of a module over a principal ideal ring
$$\Omega = T/(d_1) \times T/(d_2) \times \ldots \times T/(d_n),$$
where $T \subseteq \C$ is a principal ideal domain (PID), $d_1,d_2...,d_n \in T$ are nonzero, non-invertible elements, 
and $d_n | d_{n-1} | \ldots |d_1$. Let $R=T/(d_1)$, then $R$ is a principal ideal ring (PIR). 
The ambient space $\Omega$ is isomorphic to an $R$-submodule of $R^n$:
\begin{equation}\label{omega}
\Omega \cong R \times (d_1/d_2)\times \ldots \times (d_1/d_n)\subseteq R^n.
\end{equation}
In particular, the elements of $\Omega$ can be represented via vectors of length $n$ with entries in $R$. Both $R$ and $\Omega$ are usually finite. 

\begin{examples}
We include a list of rings $R$ and ambient spaces $\Omega$ that have been proposed in the context of physical-layer network coding,
producing efficient communication schemes.
\begin{itemize}
\item  $R=\Omega=\Z_2[i]$, proposed and studied in \cite{Zh06}, \cite{popo06}
and \cite[Example 1]{FNKS} for 
 BPSK modulation.
\item $R=\Omega=\Z_m[i]$, with $m$ a positive integer, proposed and studied in \cite{Zh06} and \cite[Example 2]{FNKS}, and known as ``$m^2$-QAM PNC scheme''.
\item $R=\Z_p[i]$, $\Omega=R^n$, where $p$ is a prime number, proposed in \cite{nazer}. 
\item $R=\Z_{p^s}[i]$, $\Omega= \underbrace{R\times\ldots\times R}_{n_1} \times \underbrace{(p) \times\ldots\times (p)}_{n_2} \times \ldots \times \underbrace{(p^{s-1})\times\ldots\times (p^{s-1})}_{n_s}$,
 where $p$ is a prime power and $s \ge 1$ is an integer, proposed in \cite[Section VII.A]{FNKS}.
\item $R=\Z[\omega]/(\beta)$, where $\Z[\omega]$ are the Eisenstein
integers and $\beta \in \Z[\omega]$ is a suitable element, $\Omega=R^n$,
proposed and studied in \cite{Sun}.
\end{itemize}
\end{examples}

As observed in \cite{FSK}, channel equation (\ref{cheq}) suggests to define the message to be transmitted as the module generated over $R$ by the rows of the matrix $X$, which we denote by $\mbox{row}(X)$. A receiver attempts to recover the original transmitted module from the matrix $Y$.

Two important special cases of equation (\ref{cheq}) correspond to the \textbf{noise-free multiplicative matrix channel} (MMC), with equation $Y = AX$,
and the \textbf{multiplication-free additive matrix channel} (AMC), with equation $Y=X+Z$. 
These two channel equations are studied in \cite{FSK} in the case where the base ring $R$ is a finite chain ring. 

The key tool in handling the MMC in \cite{FSK} is the reduced row-echelon form of a matrix $X$, which is a canonical invariant of the row-module of $X$ denoted by $\mbox{RREF}(X)$.
In practice, a transmitter emits a matrix $X$ in reduced row-echelon form, and a receiver attempts to recover it by computing $\mbox{RREF}(Y)= \mbox{RREF}(AX)$.
Decoding is successful when $A$ is left-invertible. In the same paper, the authors propose a coding/decoding scheme based on error-trapping for the AMC and the general case of a channel with equation (\ref{cheq}).

In this paper, in analogy with the approach from~\cite{KK1} for random linear network coding, we propose a new algebraic framework for module transmission based on the notion of length of a module. We define a submodule code as a collection of submodules of the ambient space $\Omega$, and propose a notion of distance between submodules based on length, which we call submodule distance. Then we show that the submodule distance captures both information loss and errors in module transmissions. The row-echelon form for modules over a PIR proposed by Buchmann and Neis in~\cite{buch} allows us to represent messages in a canonical way. Using the same row-echelon form, we reduce the computation of the distance between submodules to the computation of the length of ideals in the base ring. We also prove that, in some cases, the error-trapping decoding scheme from~\cite{FSK} is a minimum-distance decoding with respect to the submodule distance.

We derive two bounds on the cardinality of a submodule code of given minimum distance and whose codewords have fixed length. For certain classes of rings, we are able to state our bounds  explicitly in terms of the ring and code's invariants. We also construct submodule codes with maximum error-correction capability. For $R$ a finite chain ring or $R=\Z_p[i]$, we show that the codes that we construct have asymptotically optimal cardinality for their parameters. This also shows that our bounds are sharp for certain choices of rings and code parameters. We also give some general code constructions, based on the tensor and on the cartesian product. Our constructions can be applied to various choices of rings and code parameters. 

Finally, we study codes over products of rings. This is relevant, since a finite PIR $R$ is isomorphic to a product of finite fields and finite chain rings. We show that if $R\cong R_1\times\ldots\times R_m$, then a product of codes on the $R_i$'s yields a code on $R$, whose parameters are determined and whose decoding can be reduced to decoding on each of the $R_i$'s. However, not every code over a finite PIR is a product of codes over fields and finite chain rings. We give a construction of a code over $R$ which is not a product and show that decoding cannot be reduced to decoding on each of the $R_i$'s. This shows in particular how the study of codes over a finite PIR cannot be reduced to the study of codes over finite fields and finite chain rings.

The structure of the paper is as follows: In Section \ref{secalgprel} we recall some definitions and results about PIR's, modules, length, row-echelon forms of matrices over PIR's. In Section \ref{secsubmcodes} we define submodule codes and submodule distance, and relate it to information loss and errors in module transmissions. We also show how to efficiently compute the submodule distance. In Section \ref{secrecovering} we prove that, in some cases, the error-trapping decoding from~\cite{FSK} can be viewed as a minimum-distance decoding in our framework. Section \ref{secbound} is devoted to bounds on the cardinality of submodule codes, and Section \ref{seccostr} to submodule codes constructions and to codes over products of rings.

\section{Algebraic preliminaries} \label{secalgprel}

Throughout the paper $R$ denotes a finite PIR and $(r)$ denotes the ideal generated by $r\in R$. 
Recall that elements $a,b \in R$ generate the same ideal if and only if they are {\bf associate}, i.e., there exists an invertible element $\varepsilon \in R$ with $a=\varepsilon b$
(see e.g. \cite{whenare?}). An element $g \in R$ is a \textbf{greatest common divisor} (\textbf{gcd}) of $a_1,...,a_s \in R$ if and only if $(a_1) + (a_2) + \ldots + (a_s)=(g)$. We write $g=\gcd(a_1,...,a_s)$. The gcd is unique up to associates.

{\bf Finite chain rings} are a special case of PIR's. A ring $R$ is a finite chain ring if it is finite and its ideals form a chain with respect to inclusion. 
It is well-known that finite chain rings are principal and local (see e.g. \cite{macid}, page 339). Moreover, if $\pi$ is a generator of the maximal ideal of $R$, then the ideals of $R$ are 
\begin{equation}\label{fcr}
0 \subsetneq (\pi^{e-1}) \subsetneq (\pi^{e-2}) \subsetneq \ldots \subsetneq (\pi) \subsetneq R,
\end{equation}
where $e$ is the smallest positive integer with $\pi^e=0$. 
The integer $e$ does not depend on the choice of the generator $\pi$ of the maximal ideal. 
The finite field $R/(\pi)$ is called the \textbf{residue field} of $R$. Clearly, $R/(\pi) \cong \F_q$ for some prime power $q$. 

For any PIR $R$, define the {\bf annihilator} of $a\in R$ as $$\ann(a)=\{r\in R\mid ar=0\}.$$ The annihilator is an ideal of $R$, and we refer to a generator of $\ann(a)$ again as the {\bf annihilator} of $a$. If $R$ is a finite chain ring with ideal chain as in (\ref{fcr}), then every $a\in R$ is of the form $a=u\pi^\alpha$ for some invertible $u$ and some $0\leq \alpha\leq e$. Then $\ann(a)=(\pi^{e-\alpha})$. Since every finite PIR $R$ is isomorphic to a product of finite fields and finite chain rings, then annihilators are easy to compute. In the sequel, we will take computation of annihilators for granted. Moreover, inclusion of annihilators can be easily tested by checking divisibility.

\begin{proposition} \label{annihdiv}
Let $R$ be a finite PIR, $a,b\in R$. Then $\ann(a)\subseteq \ann(b)$ if and only if $a\mid b$.
\end{proposition}

\begin{proof}
By the Zariski-Samuel Theorem, any finite PIR is isomorphic to a product of finite fields and finite chain rings. Hence it suffices to prove the statement for $R$ a finite chain ring. If $b=ac$ for some $c\in R$, then $tb=tac=0$ for every $t\in\ann(a)$. Hence $\ann(a)\subseteq\ann(b)$. Conversely, if $\ann(a)\subseteq\ann(b)$ and $R$ is a finite chain ring with ideal chain as in (\ref{fcr}), then $a=u\pi^\alpha$ and $b=v\pi^\beta$ for some $u,v$ invertible, $0\leq \alpha,\beta \leq e$. Since $(\pi^{e-\alpha})=\ann(a)\subseteq\ann(b)=(\pi^{e-\beta})$, then $e-\alpha\geq e-\beta$, so $\alpha\leq\beta$ and $a\mid b$.
\end{proof}

\subsection{Modules and length}

We fix an $R$-module $\Omega\subseteq R^n$ as in (\ref{omega}) and let $\mM(\Omega)$ denote the set of $R$-submodules of $\Omega$. 
Given $M,N\in\mM(\Omega)$, denote by $M+N$ the smallest submodule of $\Omega$ which contains both $M$ and $N$. We write $M \oplus N$ when 
$N \cap M=0$. Since $R$ is finite, $\Omega$ and its submodules are finite. In particular all modules that we consider are finitely generated.

\begin{definition}
Let $M$ be an $R$-module. If $$M=\{r_1m_1+\ldots+r_tm_t\mid r_1,\ldots,r_t\in R\}$$ for some $m_1,\ldots,m_t\in M$, 
then we say that $M$ is {\bf generated} by $m_1,\ldots,m_t$ and $m_1,\ldots,m_t$ are a {\bf system of generators} for $M$. We write $M=\langle m_1,\ldots,m_t\rangle_R$
or $M=\langle m_1,\ldots,m_t\rangle$ when there is no ambiguity.
A module is {\bf finitely generated} if it has a finite system of generators.
\end{definition}

\begin{definition}
Let $M$ be an $R$-module. A chain of distinct submodules of $M$ of the form
$$0 \subsetneq M_1 \subsetneq M_2 \subsetneq \ldots \subsetneq M_\lambda=M$$
has {\bf length} $\lambda$. The \textbf{length} of $M$ is 
$$\lambda_R(M)=\max\{\lambda\mid \mbox{ $M$ has a chain of $R$-submodules of length $\lambda$}\}.$$
A {\bf composition series} for $M$ is a maximal chain of distinct $R$-submodules of $M$.
\end{definition}

When the ring is clear from context, we will omit the subscript $R$.

\begin{remark}
Any ring $R$ is an $R$-module, and its $R$-submodules coincide with its ideals. Therefore, the length of a ring $R$ is
$$\lambda(R)=\max\{\lambda\mid \mbox{ $R$ has a chain of ideals of length $\lambda$}\}.$$
For any $a\in R$ we denote by $\lambda(a)$ the length of the ideal generated by $a$. 
\end{remark}

Since all modules that we consider are finite, they have finite length.
The following properties are well-known (see e.g. \cite[Section V\S2]{kunz}).

\begin{lemma} \label{prellen}
Let $R$ be a ring, let $M,N,\Omega$ be $R$-modules of finite length. Then:
\begin{enumerate}[label={\arabic*)}]
\item $\lambda(M)=0$ if and only if $M=0$.
\item If $N \subseteq M$, then $\lambda(N) \le \lambda(M)$ and $\lambda(M/N)=\lambda(M)-\lambda(N)$. In particular, $\lambda(M)=\lambda(N)$ if and only if $M=N$.
\item If $M,N \subseteq \Omega$, then $\lambda(M+N)=\lambda(M)+\lambda(N)-\lambda(M \cap N)$.
\item $\lambda(M\times N)=\lambda(M)+\lambda(N)$.
\end{enumerate}
\end{lemma}

The concept of length for an $R$-module generalizes the concept of dimension for a vector space.

\begin{example}[fields]
Every field $\F$ is a ring of length one, since it has no proper nonzero ideals. An $\F$-module $M$ is an $\F$-vector space with 
$\lambda(M)=\dim_{\F}(M)$. 
\end{example}

\begin{example}[finite chain rings] \label{exfcr}
Let $R$ be a finite chain ring. Let $\pi$ be a generator of its maximal ideal and let $e$ be the smallest positive integer such that $\pi^e=0$. 
Then $\lambda(R)=e$ and $\lambda(a) = \min\{ i : a \in (\pi^{e-i})\} = \min \{ i : \pi^{e-i} \mid a\}$ for all $a\in R$.

It can be shown that every $R$-module $M$ is isomorphic to a product of ideals:
$$M \cong \underbrace{R \times \ldots \times R}_{\mu_1} \times \underbrace{(\pi) \times \ldots \times (\pi)}_{\mu_2-\mu_1} \times 
\ldots \times \underbrace{(\pi^{e-1}) \times \ldots \times (\pi^{e-1})}_{\mu_e-\mu_{e-1}},$$
where $0 \le \mu_1 \le \mu_2 \le \ldots \le \mu_e$ are non-negative integers. 
Following \cite{FSK}, we say that $M$ has \textbf{shape} $(\mu_1,\mu_2,...,\mu_e) \in \N^e$. 
It is easy to show that two $R$-modules are isomorphic if and only if they have the same shape.  Moreover, 
by Lemma \ref{prellen} 
$$\lambda(M)=\sum_{i=0}^{e-1} (\mu_{i+1}-\mu_i)(e-i),$$
where $\mu_0=0$.
\end{example}

%

\subsection{A reduced row-echelon form for matrices over principal ideal rings}\label{secechelon}

Every finitely generated module $M\subseteq\Omega$ may be represented as the rowspace of a matrix with entries in $R$, whose row vectors are a system of generators of $M$. In order to make such a representation unique, we need a row-canonical form for matrices with entries in a PIR.

In \cite{howell} Howell proposes a definition of row-canonical matrix over the ring $\Z_n$, showing that every matrix can be put in row-canonical form by performing invertible row operations. The ideas of Howell were later extended in \cite{buch}, where canonical generating systems for submodules of $R^n$ are defined for any PIR $R$. In the rest of the section we recall the main results of \cite{buch}, stating them in a convenient matrix formulation. 

\begin{definition}
For $i \in \{1,...,n\}$ we denote by $v_i$ the $i$-th entry of a vector $v \in R^n$. The \textbf{leading position} of a vector $v \in R^n$ is the position of its first nonzero entry:
$$\ell(v)= \left\{ \begin{array}{ll} 
                     \min \{ 1 \le j \le k \ | \ v_j
\neq 0\} & \mbox{ if $v \neq 0$,} \\
  +\infty & \mbox{ if $v=0$}.                  \end{array}
\right.\ $$
\end{definition}

Given a matrix $A \in R^{t \times n}$, we denote by $A_1,...,A_t$ the rows of $A$, and by $\mbox{row}(A)= \langle A_1,...,A_t \rangle$ the $R$-module generated by the rows of $A$.

For a module $M\subseteq R^n$ we set $M^{(j)}=\{v \in M \ | \ v_i=0 \mbox{ for } i < j \}$
for $j \in \{1,...,n+1\}$. Every $M^{(j)}$ is an $R$-submodule of $M$ and 
$$0=M^{(n+1)} \subseteq M^{(n)} \subseteq \ldots \subseteq M^{(1)}=M.$$

\begin{definition}\label{rref}
Let $A \in R^{t \times k}$ be a matrix, and let $M=\mbox{row}(A)$. We say that $A$ is in \textbf{row-echelon form} if the following hold:
\begin{enumerate}[label={\arabic*)}]
  \item for all $i \in \{1,...,t-1\}$ we have $\ell(A_{i+1}) > \ell(A_i)$,
  \item for all $j \in \{ 1,...,n+1\}$ we have $M^{(j)} = \langle A_i \ | \ \ell(A_i) \ge j \rangle$.
\end{enumerate}
The nonzero entries of $A$ of the form $A_{i,\ell(A_i)}$ are the \textbf{pivots} of $A$. 

Fix canonical generators and representatives for the ideals and the residue classes of $R$.
We say that $A$  is in \textbf{reduced row-echelon form} if it is in row-echelon form and the following hold:
\begin{enumerate}[label={\arabic*)}]
\item[3)] every pivot $A_{ij}$ of $A$ is the canonical generator of the ideal $(A_{ij})$,
\item[4)] if $A_{ij}$ is a pivot of $A$, then every entry $A_{sj}$ with $s<i$ is the canonical representative of the residue class $A_{sj}+(A_{ij})$.
\end{enumerate}
\end{definition}

\begin{definition}
Two matrices $A,B \in R^{t \times n}$ are \textbf{row-equivalent} if there exists an invertible matrix $U \in R^{t \times t}$ such that $A=UB$. 
\end{definition}

It is easy to show that $A$ and $B$ are row-equivalent if and only if $\mbox{row}(A)=\mbox{row}(B)$. We now prove that every matrix is row-equivalent to a matrix in row-echelon form, and to a unique matrix in reduced row-echelon form. The next lemma is well-known and appears in several references. We include a proof, since we could not find a complete one in the literature. 

\begin{lemma} \label{lemmapr}
For any $a,b \in R \setminus \{0\}$ there exist $x,y,z,t \in R$ such that
$$\begin{bmatrix}
x & y \\ z & t 
\end{bmatrix} 
\begin{bmatrix}
a \\ b \end{bmatrix} = 
\begin{bmatrix}
g \\ 0 
\end{bmatrix}, 
\ \ \ \ \ \  xt-yz=1, \ \ \ \ \ \ 
(g)=(a)+(b).$$
\end{lemma}

\begin{proof}
By~\cite[Theorem~33, pg. 245]{ZS}, every principal ideal ring is isomorphic to a direct product of quotients of PID's. Therefore it suffices to prove the result for $R=D/(\pi)$, where $D$ is a PID and $\pi \in D$.

If $\pi=0$ then $R=D$ is a PID. Let $g=\gcd(a,b)$ and write $g=ka+hb$ for some $k,h \in R$. Dividing by $g$ we obtain $1=k(a/g)+h(b/g)$. Let $x=k$, $y=h$, $z=-b/g$, $t=a/g$.

Now assume $\pi \neq 0$. Recall that every PID $D$ is a unique factorization domain. Consider the  projection map $\bar{\cdot}: D \to R$ and choose elements $\alpha,\beta \in D$ with $\overline{\alpha}=a$ and $\overline{\beta}=b$. Let $\eta=\mbox{gcd}(\alpha,\beta,\pi)$ and $P=\{p \in D \ | \ p \mbox{ is irreducible and divides both } \alpha/\eta
\mbox{ and } \pi/\eta  \}$. For all $p\in P$ we set $e_p=\max \{ i \ | \ p^i \mbox{ divides } \pi/\eta \}$. Define 
$$\gamma=\left\{ \begin{array}{ll} \left( \pi/\eta \right) / \prod_{p \in P}
p^{e_p} & \mbox{ if } P \neq \emptyset,  \\  \left( \pi/\eta \right) & \mbox{ if } P=\emptyset. \end{array} \right.\ $$

We claim that $\alpha/\eta+\gamma(\beta/\eta)$ and $\pi/\eta$ are coprime. By contradiction, let $p\in D$ be irreducible, $p\mid\gcd(\pi/\eta, \alpha/\eta+\gamma(\beta/\eta))$. If $p\nmid\gamma$, then $p\mid\prod_{p \in P} p^{e_p}$, so $p\mid (\alpha/\eta)$. Since $\gcd(\alpha/\eta,\beta/\eta,\pi/\eta)=1$ and $p\mid (\alpha/\eta), (\pi/\eta)$, then $p\nmid\alpha/\eta+\gamma(\beta/\eta)$, a contradiction. If $p\mid\gamma$, then $p\nmid\prod_{p \in P} p^{e_p}$, so $p\nmid(\alpha/\eta)$. However $p\mid (\alpha/\eta+\gamma(\beta/\eta))$, a contradiction. We conclude that $\alpha/\eta+\gamma(\beta/\eta)$ and $\pi/\eta$ are coprime, hence there exist  $\lambda,\mu \in D$ such that $\lambda(\alpha/\eta+\gamma(\beta/\eta)) + \mu (\pi/\eta)=1$, i.e. $\lambda(\alpha+\gamma\beta) + \mu\pi=\eta$. Hence
$$(\eta) \subseteq (\alpha+\gamma\beta,\pi) \subseteq (\alpha,\beta,\pi)=(\eta).$$
Therefore $(\alpha+\gamma\beta)+(\pi)=(\alpha)+(\beta)+(\pi)$, so $a+cb=\gcd(a,b)$ in $R=D/(\pi)$, where $c=\overline{\gamma}$. Write $b=h(a+cb)$, for some $h \in R$. 
Let $x=1$, $y=c$, $z=-h$, $t=-ch+1$.
\end{proof}

\begin{remark}
The element $\gamma\in D$ in the proof of Lemma \ref{lemmapr} can be computed (up to associates) via the following algorithm. 

\begin{algorithmic}
\State{$\gamma:=\pi/\eta$}
\State{$g:=\mbox{gcd}(\gamma,\alpha/\eta)$}
\While {$g \neq 1$}
    \State $\gamma\gets \gamma/g$
    \State $g \gets \mbox{gcd}(\gamma,\alpha/\eta)$
\EndWhile
\end{algorithmic}
\end{remark}

\begin{theorem}[\cite{buch}, Algorithm~3.2 and Theorem~3.3] \label{put}
Given a matrix $A \in R^{t \times k}$, we can compute a row-equivalent matrix in row-echelon form in $\OO(tk^2)$ operations in $R$.
\end{theorem}

We describe Algorithm~3.2 from~\cite{buch}, adapting it to our matrix notation. 
\begin{enumerate}[label={\arabic*)}]
\item If $A$ is the zero matrix, then it is already in row-echelon form. Otherwise, up to permuting the rows of $A$, we may assume without loss of generality that $+\infty > \ell(A_1) \ge \ell(A_2) \ge \ldots \ge \ell(A_t)$.
\item If $j=\ell(A_1) > \ell(A_2)$ then the first step is concluded. Otherwise, let $g= \mbox{gcd}(A_{1j},A_{2j},...,A_{tj})$. Applying Lemma \ref{lemmapr} iteratively one finds a row-equivalent matrix $A' \in R^{t \times n}$ with $A'_{1j}=g$ and $\ell(A'_i)> j$ for $i>1$.
\item Let $x \in R$ be the annihilator of $g$. We append $x \cdot A_1$ to the matrix $A'$, obtaining a matrix $A''$. Notice that $\mbox{row}(A'')=\mbox{row}(A)$. Moreover, if $v \in \mbox{row}(A'')$ and $v_s=0$ for $1 \le s \le j$, then $v \in \langle A''_2,...,A''_{t+1} \rangle$.
\end{enumerate}
One repeats the three steps above on the matrix obtained from $A''$ by deleting the first row, until there are no more rows left. The algorithm produces a matrix in row-echelon form, which is row-equivalent to $A$.

\begin{example} \label{semplice} 
Consider the matrix
$$A=\begin{bmatrix}
     2 & 1 & 3 \\ 4 & 1 & 2
    \end{bmatrix}$$
over $R=\Z_6$. Applying the algorithm that we just described, one computes:
$$\begin{bmatrix} 2 & 1 & 3 \\ 4 & 1 & 2 \end{bmatrix}
\rightsquigarrow    
 \begin{bmatrix} 2 & 1 & 3 \\ 0 & 5 & 2 \end{bmatrix}
 \rightsquigarrow    
 \begin{bmatrix} 2 & 1 & 3 \\ 0 & 5 & 2  \\ 0 & 3 & 3 \end{bmatrix}
 \rightsquigarrow    
 \begin{bmatrix} 2 & 1 & 3 \\ 0 & 5 & 2  \\ 0 & 0 & 3 \end{bmatrix}.
 $$
Notice that the number of rows increased from two to three. 
\end{example}

Fix generators and representatives for the ideals and residue classes of $R$. 
Every matrix $A$ with entries in $R$ is row-equivalent to a unique matrix in reduced row-echelon form, with respect to the given choice of generators and representatives. 

\begin{theorem}[\cite{buch}, Algorithm~3.4 and Theorem~3.5] \label{uniqueness}
Let $A \in R^{t \times n}$ be a matrix. Then $A$ is row-equivalent to a unique matrix in reduced row-echelon form, which we denote by $\rref(A)$. 
The reduced row-echelon form of $A$ can be computed from a row-echelon form of $A$ in $\OO(n^3)$ operations in $R$.
\end{theorem}

In fact, using the algorithm from Theorem~\ref{put}, $A$ can be put in row-echelon form. After multiplying it by a diagonal matrix with invertible elements on the diagonal, the matrix satisfies property 3) of Definition \ref{rref}. Finally, by subtracting suitable multiples of each row from the rows above, one ensures that property 4) of Definition \ref{rref} holds. The last operation corresponds to multiplying on the left by an upper triangular matrix with ones on the diagonal. 

The next result characterizes matrices in row-echelon form over a finite PIR. A proof can be obtained using Proposition \ref{annihdiv}.

\begin{proposition} \label{criterio}
Let $A \in R^{t \times n}$ be a matrix with no zero rows. $A$  is in row-echelon form if and only if the following hold:
\begin{enumerate}
\item  $\ell(A_{i+1}) > \ell(A_i)$ for all $i \in \{1,...,t-1\}$,
\item $A_{t\ell(A_t)} \mid A_{tj}$ for all 
$\ell(A_t)\le j \le n$,
\item  $\mbox{ann}(A_{i\ell(A_i)}) \cdot A_i \in 
\langle A_{i+1},...,A_t \rangle$ for all $i \in \{1,...,t-1\}$.
\end{enumerate}
\end{proposition}

Proposition \ref{criterio} and Algorithm 4.1 of \cite{buch} lead to the following algorithm to test whether a matrix is in row-echelon form.

\vspace{0.3cm}

\begin{algorithmic}
\State{\textbf{Input}: a matrix $A \in R^{t \times n}$ with $\ell(A_{i+1}) > \ell(A_i)$ for all
$i \in \{1,...,t-1\}$}
\State{\textbf{Output}: ``YES" if $A$ is in row-echelon form, and ``NO" otherwise}
\For{$i=1 \ {\bf{to}} \ t$}
\State{$j_i:=\ell(A_i)$}
\EndFor
\For{$j=j_t \ {\bf{to}} \ n$}
\If{$A_{tj_t} \nmid A_{tj}$}
\State{\Return{NO}}
\State{\textbf{quit}}
\EndIf
\EndFor
\For{$i=t-1 \ {\bf{downto}} \ 1$}
\State{use Algorithm 4.1 of \cite{buch} to test if 
$\mbox{ann}(A_{ij_i}) \cdot A_i \in \langle A_{i+1},...,A_t\rangle$}
\If{$\mbox{ann}(A_{ij_i}) \cdot A_i \notin \langle A_{i+1},...,A_t\rangle$}
\State{\Return{NO}}
\State{\textbf{quit}}
\EndIf
\EndFor
\State{\Return{YES}}
\State{\textbf{quit}}
\end{algorithmic}

\begin{proposition}
The previous algorithm terminates correctly.
\end{proposition}

\begin{proof} Algorithm 4.1 of \cite{buch} tests whether a given vector belongs to the module generated by the rows of a matrix in row-echelon form. Therefore, we first need to show that the last \textbf{for} cycle of the algorithm is well-defined, i.e., that if the algorithm enters the \textbf{for} cycle for some $i$, then the matrix whose rows are $A_{i+1},...,A_t$ is in row-echelon form.

We proceed by backward induction on $i \in \{t-1,...,1\}$. Assume that the algorithm enters the cycle for $i=t-1$. Then by 
the structure of the algorithm and Proposition \ref{criterio}, the matrix whose row is $A_t$ is in row-echelon form, as claimed.
Now assume $i<t-1$. Since the algorithm enters the \textbf{for} cycle for $i$, it entered the 
\textbf{for} cycle also for $i+1$. By induction hypothesis, the matrix whose rows are 
$A_{i+2},...,A_t$ is in row-echelon form. Since the algorithm enters  the \textbf{for} cycle for $i$,
we have $\mbox{ann}(A_{i+1}) \cdot A_{i+1} \in \langle A_{i+2},...,A_t \rangle$.
Therefore by Proposition \ref{criterio} the matrix whose rows are 
$A_{i+1},...,A_t$ is in row-echelon form.

The previous argument also shows that if the algorithm returns YES, then $A$ is in row-echelon form.
Finally, using Proposition \ref{criterio} one can check that if $A$ is in row-echelon form, then
the algorithm returns YES.
\end{proof}

\section{Submodule codes and submodule distance} \label{secsubmcodes}

Using the length, one can define a distance function between submodules of $\Omega$.

\begin{proposition}\label{dist}
The function $d: \mM(\Omega) \times \mM(\Omega) \to \N$ defined by $$d(M,N)=\lambda(M)+\lambda(N)-2\lambda(M \cap N)$$ for all $M,N \in \mM(\Omega)$ is a distance function.
\end{proposition}

\begin{definition}\label{defsubdist}
We call $d$ the \textbf{submodule distance} on $\mM(\Omega)$.
\end{definition}

\begin{proof}[Proof of Proposition~\ref{dist}]
Let $M,N,P\subseteq\Omega$ be $R$-submodules. By Lemma~\ref{prellen} we have $d(M,N)=\lambda(M+N)-\lambda(M\cap N)$.
Since $M \cap N \subseteq M +N$ we have $d(M,N) \ge 0$, and equality holds if and only if $M+N=M \cap N$, i.e., if and only if $M=N$. 
Moreover, $d(M,N)=d(N,M)$ by definition. 

To prove the triangular inequality, observe that by definition 
$$d(M,N)=d(M,P)+d(P,N) -2(\lambda(M \cap N) + \lambda(P) - \lambda(M \cap P) -
\lambda(N \cap P)).$$
Therefore it suffices to prove that $x=\lambda(M \cap N) + \lambda(P) - \lambda(M \cap P) - \lambda(N \cap P) \ge 0.$
By Lemma~\ref{prellen} we have $x=\lambda(M+P) + \lambda(N+P) - \lambda(M+N) -
\lambda(P)$. Since $(M+P)+(N+P) \supseteq M+N$ and $(M+P)\cap(N+P) \supseteq P$, by Lemma~\ref{prellen} 
$$\lambda(M+P) + \lambda(N+P) -\lambda(P) \ge \lambda(M+P) + \lambda(N+P) -\lambda((M+P)\cap(N+P)) \ge \lambda(M+N),$$
hence $x \ge 0$.
\end{proof}

When $R=\F$ is a field, the submodule distance on $\F^n$ coincides with the {\bf subspace distance} proposed by K\"otter and Kschischang in \cite{KK1} for error correction in random linear network coding. 

The concepts of information loss and error from~\cite{KK1} can be extended to our setting as follows.

\begin{remark} \label{remdecomp}
Let $M \subseteq R^n$ be the transmitted module, and let $N\subseteq R^n$ be the received module. 
The portion of information that was correctly transmitted is $M \cap N$. The quotient $M/(M \cap N)$ may be regarded as the \textbf{information loss module}, i.e. the original information modulo the portion of information that was correctly transmitted. 
Similarly, the \textbf{error module} is the quotient $N/(M \cap N)$. Using Lemma~\ref{prellen}, one can check that 
$$d(M,N)= \lambda(M/(M \cap N)) + \lambda(N/(M \cap N)).$$ 
In other words, the distance between $M$ and $N$ is the sum of the lengths of the information loss module and of the error module, similarly to what was shown in~\cite{KK1} in the context of subspace codes. 
\end{remark}

\begin{definition}
A \textbf{submodule code} is a subset $\mC \subseteq \mM(\Omega)$ with $|\mC| \ge 2$.
The \textbf{minimum} (\textbf{submodule}) \textbf{distance} of $\mC$ is $$d(\mC)= \min \{d(M,N) : M,N \in \mC, \ M \neq N\}.$$
\end{definition}

\begin{definition}
Let $\mC \subseteq \mM(\Omega)$ be a submodule code. Let $M \in \mC$ be the transmitted module, and let $N \in \mM(R^n)$ be the received module.
Define the \textbf{number of erasures} as $\rho=\lambda(M/(M \cap N))$ and the \textbf{number of errors} as $e=\lambda(N/(M \cap N))$. 
\end{definition}

The next result follows from Remark~\ref{remdecomp} using a standard argument.

\begin{proposition}
Let $\mC \subseteq \mM(\Omega)$ be a submodule code of minimum distance $d$.
Then a minimum distance decoder successfully corrects $N$ to $M$, provided that $2(\rho+e) < d(\mC)$.
\end{proposition}

\subsection{Computing the distance function}

In this subsection we show that the length of a module is the sum of the lengths of the ideals generated by the pivots of a matrix in row-echelon form, whose rows generate the module. Therefore, computing the length of an $R$-module can be reduced to computing lengths of ideals in $R$. This allows us to efficiently compute distances between submodules of $R^n$.

\begin{theorem} \label{princ}
Let $M\subseteq R^n$ be an $R$-module. Let $A\in R^{t \times n}$ be a matrix in row-echelon form with no zero rows and such that $M=\mbox{row}(A)$. 
For every $j \in \{ 1,...,n+1\}$ let $M^{(j)}=\{m\in M\mid v_i=0 \mbox{ for } i<j\}\subseteq R^n$ and $I^{(j)}= (v_j \ | \ v \in M^{(j)}) \subseteq R$. 
Let $A_{i\ell(A_i)}$ be the pivot of the $i$-th row of $A$.
Then: 
\begin{enumerate}[label={\arabic*)}]
\item $I^{(\ell(A_i))}=\left( A_{i\ell(A_i)} \right) \cong M^{(j)}/M^{(j+1)}$ for all $i \in \{1,...,t\}$,
\item \label{princ3} $\lambda(M)= \sum_{i=1}^t \lambda \left(A_{i \ell(A_i)}\right)$.
\end{enumerate}
\end{theorem}

\begin{proof} 
\begin{enumerate}[label={\arabic*)}]
\item The map $M^{(j)}/M^{(j+1)} \to I^{(j)}$ given by $v+M^{(j+1)} \mapsto v_j$ is a well-defined $R$-module isomorphism. Therefore $I^{(j)} \cong M^{(j)}/M^{(j+1)}$.
Fix any $i \in \{1,...,t\}$ and let $j=\ell(A_i)$. Since $A_i \in M^{(j)}$ we have $(A_{ij}) \subseteq I^{(j)}$. On the other hand, let $0\neq x \in I^{(j)}$ and let $v \in M^{(j)}$ such that $x=v_j$. Since $A$ is in row-echelon form, $v=\sum_{k=i}^t r_k A_k$ for some $r_i,...,r_t\in R$. Therefore $x=v_j=r_i A_{ij}$, so $x \in (A_{ij})$.
\item Applying \cite[Corollary V.2.4]{kunz} to the chain of $R$-modules 
$0=M^{(n+1)} \subseteq \ldots \subseteq M^{(1)}=M$ we obtain 
$$\lambda(M)=\sum_{j=1}^n \lambda\left(M^{(j)}/M^{(j+1)}\right).$$
By Lemma~\ref{prellen} one has $\lambda\left(M^{(j)}/M^{(j+1)}\right) \neq 0$ if and only if $I^{(j)}\neq 0$, if and only if $j=\ell(A_i)$ for some $i\in\{1,\ldots,t\}$. Then 
$$\lambda\left(M^{(j)}/M^{(j+1)}\right) 
= \lambda(A_{ij}).$$\qedhere
 \end{enumerate}
\end{proof}

\begin{example} 
The module $M=\langle(2,1,3), (4,1,2)\rangle \subseteq\Z_6^3$ generated by the rows of the matrix of Example \ref{semplice} has length $\lambda(M)=\lambda(2) + \lambda(5) + \lambda(3)= 1+2+1=4$.
\end{example}

\begin{remark}
Let $M=\mbox{row}(A)$ be the transmitted module, and let $N=\mbox{row}(B)$ be the received module. By Lemma \ref{prellen} we have $$d(M,N)=2 \lambda(M+N)- \lambda(M)-\lambda(N).$$ Therefore the distance between $M$ and $N$ can be computed from the row-echelon forms of $A$, $B$, and of the matrix $C$ obtained by appending the rows of $B$ to $A$. In fact $$M+N=\mbox{row}(A)+\mbox{row}(B)=\mbox{row}(C).$$
This allows us to compute the distance function without computing intersections of modules.
\end{remark}

\begin{example}
Let $R=\Z_4$. Notice that the only nonzero, proper ideal of $R$ is $(2)$. Let 
$$A= \begin{bmatrix}
\underline{1} & 1 & 1 & 0 \\ 
0 & \underline{2} & 1 & 2 \\ 
0 & 0 & \underline{2} & 0
\end{bmatrix}, \  \ \ \ \ \ \ \ \  
B= \begin{bmatrix}
\underline{1} & 3 & 0 & 2 \\ 0 & 0 & \underline{1} & 0
\end{bmatrix}$$
be matrices in row-echelon form, whose underlined entries are the pivots.
Let $M=\mbox{row}(A)$ and $N=\mbox{row}(B)$. Then 
\begin{eqnarray*}
 \lambda(M) &=& \lambda(1) + \lambda(2) + \lambda(2) = 2+1+1=4 \\
 \lambda(N) &=& \lambda(1) + \lambda(1) = 2+2=4 
\end{eqnarray*}
Then $M+N=\mbox{row}(C)$, where
$$C=\begin{bmatrix}
1 & 1 & 1 & 0 \\ 
0 & 2 & 1 & 2 \\ 
0 & 0 & 2 & 0 \\
1 & 3 & 0 & 2 \\ 
0 & 0 & 1 & 0
\end{bmatrix}$$
whose reduced row-echelon form is
$$\begin{bmatrix}
\underline{1} & 1 & 0 & 0 \\ 0 & \underline{2} & 0 & 2 \\ 0 & 0 & \underline{1} & 0
\end{bmatrix}.$$
Hence $\lambda(M+N)= \lambda(1) + \lambda(2) + \lambda(1) = 2+1+2=5$ and
$$d(M,N)=2\lambda(M+N) - \lambda(M) - \lambda(N)=10-8=2.$$
One can compute $\lambda(M\cap N)=\lambda(M)+\lambda(N)-\lambda(M+N)=3$.
Hence the information loss module has length $\lambda(M/(M \cap N))=\lambda(M)-\lambda(M \cap N)=1$ and the error module has length 
$\lambda(N/(M \cap N))=\lambda(N)-\lambda(M \cap N)=1$.
\end{example}

\section{Recovering known encoding and decoding schemes}\label{secrecovering}

In this section we compare our approach to the one proposed by Feng, N\'obrega, Kschischang, and Silva for the multiplicative-additive matrix channel (MAMC) in \cite[Section IX]{FNKS}. We show that their encoding scheme remains valid in our setup. We also prove that, in some cases, their decoding scheme corresponds to minimum distance decoding with respect to the distance function that we propose.

In our notation, Feng, N\'obrega, Kschischang, and Silva consider the MAMC of equation $Y=AX+Z$, where $R$ is a finite chain ring, $A\in R^{N\times t}$, $X\in R^{t\times n}$, $Z\in R^{N\times n}$. They assume that $n\geq 2N$, $\row(A)\cong R^t$, and $\row(Z)\cong R^v$ for some integer $v\leq N$. They represent matrices in {\bf row canonical form} (see~\cite[Definition 1]{FNKS} for the definition of row canonical form) and define their codebook to be the set of \textbf{principal} matrices of given shape in row canonical form (see~\cite[Section V.B and Sections VII, VIII, IX]{FNKS}). Observe that matrices in row canonical form are not in general in reduced row-echelon form according to our Definition~\ref{rref}.  

\begin{example}
Let $R=\Z_8$, whose ideals $0\subsetneq (4)\subsetneq (2)\subsetneq (1)$ have canonical generators $0,4,2,1$. Choose $0,1,2,3$ as canonical representatives for residue classes modulo $(4)$, $0,1$ as canonical representatives for residue classes modulo $(2)$, and $0$ as canonical representative for the residue class modulo $(1)$.
The rows of the following matrices generate the same $R$-module. The first matrix is in row canonical form (see~\cite[Example~6]{FNKS}), while the second is in reduced row-echelon form. The pivots are underlined.
$$\left[\begin{matrix}
0 & 2 & 0 & \underline{1} \\
\underline{2} & 2 & 0 & 0 \\
0 & 0 & \underline{2} & 0 \\
0 & \underline{4} & 0 & 0
\end{matrix}\right] \;\;\;\;\;\;\;\;
\left[\begin{matrix}
\underline{2} & 0 & 0 & 1 \\
0 & \underline{2} & 0 & 1 \\
0 & 0 & \underline{2} & 0 \\
0 & 0 & 0 & \underline{2} 
\end{matrix}\right].$$ 
\end{example}

The authors of \cite{FNKS} propose asymptotically optimal encoding and decoding schemes using principal matrices over a finite chain ring $R$. A transmitted module $M$ is encoded as a principal matrix in row canonical form, whose rows generate $M$. In the next proposition we show that principal matrices in row canonical form are in reduced row-echelon form. Therefore, the encoding schemes of \cite{FNKS} remain valid in our setup.

\begin{proposition}\label{standardisstandard}
Let $R$ be a finite chain ring, and let $A \in R^{t \times n}$ be a principal matrix in row canonical form. Then $A$ is in reduced row-echelon form with respect to the same choices of generators and representatives.
\end{proposition}

\begin{proof}
Let $\pi$ be a generator of the maximal ideal of $R$, let $e$ be the smallest positive integer such that $\pi^e=0$. 
By the definition of principal row canonical form, the pivot of row $A_i$ is $A_{ii}=\pi^{\ell_i}$ for $i\in\{1,\ldots,t\}$, with $0\leq \ell_1\leq\ldots\leq\ell_t\leq e$. 
Moreover, $A_{ij}=\pi^{a_{ij}}$ with $a_{ij}\geq \ell_i$ for $i<j\le n$, hence $\ann(A_{ij}) \supseteq \ann(A_{ii})$.

Let $v \in \mbox{row}(A)\setminus\{0\}$, let $j=\ell(v)$. Write $v=\sum_{i=1}^t r_iA_i$ for some $r_1,...,r_t \in R$. Hence $v_1=r_1A_{1,1}=0, v_2=r_1A_{1,2}+r_2A_{2,2}=0,\ldots, v_{j-1}=r_1A_{i,j-1}+\ldots+r_{j-1}A_{j-1,j-1}$. By induction on $i$ one can show that $r_iA_{ii}=0$  for $1 \le i <j$, hence $r_iA_i=0$ for $1\leq i<j$.
Therefore $v=\sum_{i=j}^t r_iA_i \in \langle A_j,\ldots,A_t \rangle$, so $A$ is in reduced row-echelon form according to Definition~\ref{rref}. 
\end{proof}

We conclude this section with Proposition~\ref{decoding}, that shows that the error-trapping decoding scheme proposed in \cite{FNKS} for the MAMC can be interpreted as a minimum distance decoding with respect to the distance function from Definition~\ref{defsubdist}. Before stating our result, we recall the scheme of \cite[Section IX]{FNKS}.

\begin{example}[Error-trapping decoding]\label{extrapp}
Let $R$ be a finite chain ring. Fix $N$ such that $n\geq 2N$ and consider the channel equation $Y=AX+Z$, where $A \in R^{N \times t}$ is left-invertible, $X\in R^{t \times n}$ is the matrix whose rows generate the transmitted module, and $Z \in R^{N \times n}$ is a noise matrix whose row-module is isomorphic to $R^v$ for some integer $v\leq N$. One can write 
$$A=P \begin{bmatrix}
0_{(N-t) \times t} \\ I_t
\end{bmatrix},$$
where $P \in R^{N \times N}$ is an invertible matrix. Fix $u\geq v$. If $t+v>N$ let $X\in R^{t\times n}$ be of the form 
$$X= \begin{bmatrix}
0 & 0 \\ 0 & \overline{X}
\end{bmatrix},$$
where $\overline{X}$ is a matrix in principal form of size $(N-u) \times (n-u)$. If $t+v\leq N$ let $X\in R^{t\times n}$ be of the form 
$$X= \begin{bmatrix}
0 & \overline{X}
\end{bmatrix},$$
where $\overline{X}$ is a matrix in principal form of size $t\times (n-u)$. Under the assumption that error trapping is successful, \cite[Section IX.B]{FNKS} shows that the row canonical form of $Y=AX+Z$ is 
$$\begin{bmatrix}
   Z_1 & Z_2 \\ 0 & \overline{X} \\ 0 & 0
  \end{bmatrix},$$
for suitable matrices $Z_1\in R^{v\times u}$ and $Z_2\in R^{v\times (n-u)}$. 
Hence $\overline{X}$ and $X$ can be obtained by computing the row canonical form of $Y$. 
\end{example}

In some cases, the error-trapping decoding from~\cite{FNKS} can be seen as an instance of minimum distance decoding according to our definition.
Notice that the choice $u=v$ is particularly interesting, since it maximizes the number of codewords for the given channel, without affecting the error-correction capability of the code.

\begin{proposition}\label{decoding}
Following the notation of Example~\ref{extrapp}, and under the assumption that either $t+v=N$ or $u=v$ and $t+v>N$, we have 
$$d \left( \mbox{row}(Y), \ \mbox{row}\begin{bmatrix}
0 & 0 \\ 
0 & \overline{T}
\end{bmatrix}\right) \ \ge \ 
d \left( \mbox{row}(Y), \ \mbox{row}\begin{bmatrix}
0 & 0 \\ 
0 & \overline{X}
\end{bmatrix}\right)$$
for any principal matrix in row canonical form $\overline{T}$ of the same size as $\overline{X}$. Moreover, equality holds {if and only if} $\overline{T}=\overline{X}$.
\end{proposition}

\begin{proof}
Since error-trapping is successful, by \cite[Section IX.B]{FNKS} there exist matrices $G,H,K$ such that 
$$Y=G \cdot \begin{bmatrix}
             H & K \\ 0 & \overline{X}
            \end{bmatrix},$$
where $G \in R^{N \times N}$ is invertible, $H \in R^{v \times u}$ and $\row(H)\cong R^v$, $K \in R^{v \times (n-u)}$. Since $\row(H)\cong R^v$,
$$\mbox{row} \begin{bmatrix}
                                   H & K \\ 0 & \overline{X}
                                   \end{bmatrix}
                                   \cap \ 
\mbox{row} \begin{bmatrix} 0 & 0 \\ 0 & \overline{T}
\end{bmatrix} = 
\mbox{row} \begin{bmatrix}
                                   0 & 0 \\ 0 & \overline{X}
                                   \end{bmatrix}
                                   \cap \ 
\mbox{row} \begin{bmatrix} 0 & 0 \\ 0 & \overline{T}      
\end{bmatrix}.$$
Therefore 
$$d \left( \mbox{row}(Y), \ \mbox{row}\begin{bmatrix}
0 & 0 \\ 0 & \overline{T}
\end{bmatrix}\right) =  \lambda(\row(Y)) + \lambda(\row(\overline{T}))- 2 \lambda(\row(\overline{X})\cap \row(\overline{T}))$$
$$ = d \left( \mbox{row}(Y), \ \mbox{row}\begin{bmatrix}
0 & 0 \\ 0 & \overline{X}
\end{bmatrix}\right) + \lambda(\row(\overline{T})) + \lambda(\row(\overline{X})) - 2 \lambda(\row(\overline{X})\cap \row(\overline{T}))$$
$$ =  d \left( \mbox{row}(Y), \ \mbox{row}\begin{bmatrix}
       0 & 0 \\ 0 & \overline{X}
      \end{bmatrix}\right) + d \left( \mbox{row} \begin{bmatrix}
      0 & 0 \\ 0 & \overline{X}\end{bmatrix}, 
      \mbox{row} \begin{bmatrix}
      0 & 0 \\ 0 & \overline{T}\end{bmatrix}
      \right).$$
The result now follows from the fact that both $\overline{X}$ and $\overline{T}$ are principal matrices in row canonical form, hence they are in reduced row-echelon form by Proposition~\ref{standardisstandard}.
\end{proof}

\section{Bounds}\label{secbound}

In this section we derive two upper-bounds on the cardinality of a submodule code with given minimum distance. We also discuss in detail some choices of rings or of the code parameters, for which our bound can be made more precise. 

As in the previous sections, we fix a finite PIR $R$, an R-module $\Omega\subseteq R^n$ of the form (\ref{omega}), and let $\mM(\Omega)$ denote the set of $R$-submodules of $\Omega$. 

\begin{notation}
For $M\in\mM(\Omega)$ and $1\leq s\leq\lambda(M)$ let
$$\begin{bmatrix} M \\ s \end{bmatrix}_R= \ |\{ N\in\mM(M) : \ \lambda(N)=s\}|.$$
For $1\leq s\leq\lambda$ let 
$$\begin{bmatrix} \lambda \\ s\end{bmatrix}_R= \min\left\{  \begin{bmatrix} M \\ s\end{bmatrix}_R : M\in\mM(\Omega), \ \lambda(M)=\lambda \right\}.$$
When there is no ambiguity, we omit the subscript $R$. Moreover, we denote by 
$$\qbin{\lambda}{s}{q}=\qbin{\lambda}{s}{\F_q}$$
the $q$-ary binomial coefficient.
\end{notation}

We restrict our attention to submodules codes $\mC \subseteq \mM(\Omega)$  where all codewords have the same length $k$, $1 \le k \le n\cdot \lambda(R)-1$. Submodule codes of this kind have even minimum distance, and they are the module-analogue of constant-dimension subspace codes. The next result is a natural extension to submodule codes of the Singleton-like bound for subspace codes.
 
\begin{theorem} \label{sss}
Let $\mC \subseteq \mM(\Omega)$ be a submodule code with $\lambda(M)=k$ for all $M\in\mC$ and minimum distance $d(\mC)=2\delta$. Then
$$|\mC| \le \begin{bmatrix} \lambda(\Omega) -\delta +1 \\ k-\delta+1 \end{bmatrix}.$$
\end{theorem}

\begin{proof}
Let $M\in\mM(\Omega)$ be an $R$-module with $\lambda(M)=\lambda(\Omega) -\delta +1$. By Lemma \ref{prellen}, for all $N \in \mC$
$$\lambda(M \cap N) = \lambda(M) + \lambda(N) - \lambda(M+N) \ge \lambda(\Omega) -\delta +1 +k - \lambda(\Omega)=k-\delta +1.$$
For every $N \in \mC$ choose an $R$-submodule $N' \subseteq M \cap N$ with $\lambda(N')= k -\delta +1$. For any $N,P \in \mC$ with $N \neq P$ we have 
 $2\delta =d(\mC) \le d(N,P) = 2k-2 \lambda(N \cap P)$, hence $\lambda(N \cap P) \le k-\delta$.
Hence $$d(N',P') = 2(k-\delta+1) - 2\lambda(N' \cap P') \ge 2(k-\delta+1) -
2(k-\delta)=2,$$
in particular $N'\neq P'$. It follows that $\mC'= \{N' : N \in \mC\}$ is a set of submodules of $M$ of length $k-\delta +1$ with $|\mC'| = |\mC|$. Therefore
$$|\mC| = |\mC'| \le \begin{bmatrix} M \\ k-\delta+1 \end{bmatrix},$$
for any $M\in\mM(\Omega)$ of length $\lambda(\Omega) -\delta +1$.
\end{proof}

\begin{remark}
For a given $R$ and fixed $m,\ell$, the quantity $\begin{bmatrix} M \\ \ell \end{bmatrix}$ may depend on the choice of $M$ of length $m$. E.g., let $R=\Z_5[i]\supseteq I=(2+i)$ and $\Omega=R^2$. Then $R \times 0$ and $I\times I$ are two $R$-modules of length $2$. The ideals of length one of $R$ are exactly $I=(2+i)$ and $(2-i)$, while $I\times I$ contains at least three submodules of length one, namely $I\times 0, 0\times I,$ and $(1,1)I=\langle(2+i,2+i)\rangle$.
\end{remark}

\begin{remark}
Since every $R$-module of length greater than or equal to $\lambda(\Omega) -\delta +1$ contains an $R$-submodule of length $\lambda(\Omega) -\delta +1$, the bound of Theorem~\ref{sss} can also be stated as $$|\mC| \le \min \left\{ \begin{bmatrix} M \\ k-\delta+1 \end{bmatrix}  \ : \ M \in \mM(\Omega), \ \lambda(M)\geq\lambda(\Omega) -\delta +1\right\}.$$
\end{remark}

The following is another simple bound for the cardinality of a submodule code.

\begin{theorem} \label{bound2}
Let $\mC \subseteq \mM(\Omega)$ be a submodule code with 
$\lambda(M)=k$ for all $M \in \mC$, and minimum distance $d(\mC)=2\delta$.
Then
$$|\mC| \le \frac{\begin{bmatrix} \Omega \\ k-\delta+1\end{bmatrix}}{\begin{bmatrix} k \\ k-\delta+1\end{bmatrix}}.$$ 
\end{theorem}

\begin{proof}
Each $M \in \mC$ contains at least 
$\begin{bmatrix} k \\ k-\delta+1\end{bmatrix}$ submodules of $\Omega$ of length $k-\delta+1$. 
Moreover, a submodule of $\Omega$ of length $k-\delta+1$ cannot be contained in two distinct $M,N \in \mC$, as otherwise $\lambda(M \cap N) \ge k-\delta+1$, hence $d(M,N) < 2\delta$. Therefore
$$\begin{bmatrix} \Omega \\ k-\delta+1\end{bmatrix}\ge |\mC| \cdot \begin{bmatrix} k \\ k-\delta+1\end{bmatrix},$$
which proves the bound. 
\end{proof}

\begin{remark}
The upper bounds of Theorem \ref{sss} and \ref{bound2} are not comparable. For example, let $k=\delta$,
$R=\F_q$, and $\Omega=\F_q^n$. Assume that $k\mid n$. The bound of Theorem~\ref{sss} is 
$$|\mC|\leq \qbin{n-k+1}{1}{q}=q^{n-k}+q^{n-k-1}+\ldots+q+1,$$
while the bound of Theorem~\ref{bound2} is
$$|\mC|\leq\frac{\qbin{n}{1}{q}}{\qbin{k}{1}{q}}=\frac{q^n-1}{q^k-1}=q^{n-k}+q^{n-2k}+\ldots+q^k+1.$$
However, one can also find examples in which Theorem \ref{sss} yields a better bound than Theorem \ref{bound2}. 
E.g., let $R=\Z_{12}$. The Hasse diagram of the ideals of $R$ is
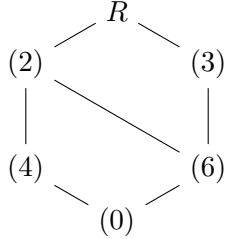
\begin{center}
 \begin{tikzpicture}[scale=.7]
  \node (one) at (90:2cm) {$R$};
  \node (b) at (150:2cm) {$(2)$};
  \node (a) at (210:2cm) {$(4)$};
  \node (zero) at (270:2cm) {$(0)$};
  \node (c) at (330:2cm) {$(6)$};
  \node (d) at (30:2cm) {$(3)$};
  \draw (zero) -- (a) -- (b) -- (one) -- (d) -- (c) -- (zero);
\draw  (b) -- (c);
\end{tikzpicture}
\end{center}
In particular, $\lambda(R)=3$. Let $\Omega=R^2$ and let $\mC \subseteq \mM(\Omega)$ be a submodule code with $k=\delta=2$. 
By Theorem \ref{princ}, the module 
$$M = \mbox{row} \begin{bmatrix} 1& 0 \\ 0 & 3 \end{bmatrix}\subseteq \Omega$$
has $\lambda(M)=\lambda(1)+\lambda(3)=5$. Moreover, the submodules of 
$M$ of length 1 are precisely those generated by one of the following vectors: $(4,0), (6,0), (6,6), (0,6).$ Therefore, $|\mC| \le 4$ by Theorem~\ref{sss}. 
Now let $N=\langle (0,3) \rangle \subseteq \Omega$. Then $\lambda(N)=2$ and $N$ has a unique submodule of length 1, namely $\langle(0,6)\rangle$. Hence
$$\qbin{2}{1}{} = \min\left\{ \qbin{N}{1}{} : N \in \mM(\Omega), \ \lambda(N)=2
  \right\} =1.$$
One can check that the submodules of $\Omega$ of length 1 are exactly those generated by one of the following vectors: $(4,0)$, $(4,4)$, $(4,8)$, $(6,0)$, $(6,6)$, $(0,4)$, $(0,6)$.
Therefore the bound of Theorem~\ref{bound2} reads $$|\mC| \le \qbin{\Omega}{1}{} / \qbin{2}{1}{}=7 > 4.$$
\end{remark}

The bounds of Theorem~\ref{sss} and Theorem~\ref{bound2} can be made more explicit for PIR's which are isomorphic to $\Z_p^m$ for some $m\geq 1$. 
An example of such ring is $\Z_p[i]$, which is isomorphic to $\Z_p^2$ if $p\equiv 1\mod 4$, as we show next. Notice that in the other cases $\Z_p[i]$ is either a finite chain ring or a finite field.

\begin{remark}\label{strui}
Let $p$ be a prime. Then $$\Z_p[i] \cong \Z_p[x]/(x^2+1)\cong \left\{  \begin{array}{ll}
\Z_2[x]/(x+1)^2  & \mbox{ if $p=2$,} \\
\F_{p^2} & \mbox{ if $p \equiv 3 \mod 4$,} \\
\Z_p \times \Z_p & \mbox{ if $p \equiv 1 \mod 4$.}                          
                         \end{array} \right.$$
Indeed, if $p=2$, then $x^2+1=(x+1)^2$. If $p$ is an odd prime, then $x^2+1$ is reducible if and only if $-1$ is a quadratic residue modulo $p$, if and only if $p \equiv 1 \mod 4$. The thesis now follows from the Chinese Remainder Theorem.
\end{remark}

We start with a preliminary result on the structure of rings of the form $R \cong \Z_p^m$.

\begin{lemma}\label{eis}
Let $R\cong\Z_p^m$. Then there exist $e_1,\ldots,e_m\in R$ such that $$R=(e_1)\oplus\ldots\oplus(e_m)$$ and $e_1+\ldots+e_m=1$, $e_i^2=e_i$ for all $1\leq i\leq m$, and $e_ie_j= 0$ if $i\neq j$. Moreover, $\lambda(R)=m$.
\end{lemma}

\begin{proof}
Fix an isomorphism $\pi : R \to \Z_p^m$ between $R$ and $\Z_p^m$.
Let $e_i\in R$ be the inverse image via $\pi$ of the $i$-th element of the standard basis of $\Z_p^m$. Then $e_1+\ldots+e_m=1$, $e_i^2=e_i$ for all $1\leq i\leq m$, and $e_ie_j= 0$ if $i\neq j$. Notice that $R$ is a $\Z_p$-vector space via $\alpha r=\pi^{-1}(\alpha\pi(r))$ for $\alpha\in\Z_p, r\in R$. This corresponds to identifying $\Z_p$ with $\{\pi^{-1}(\alpha,\ldots,\alpha) : \ \alpha\in\Z_p\} \subseteq R$.
Since $R=(e_1)\oplus\ldots\oplus (e_m)$, then every $r\in R$ can be written uniquely as $r=r_1e_1+\ldots+r_me_m$ with $r_i\in\Z_p$, where we regard $\Z_p$ as a subset of $R$ via the identification above. 
Therefore $\lambda(R)=m$ and a composition series for $R$ is given by 
$0\subsetneq (e_1)\subsetneq (e_1,e_2)\subsetneq\ldots\subsetneq (e_1,\ldots,e_m)=R.$
\end{proof}

Using the notation of Lemma~\ref{eis}, we can count the number of submodules of fixed length of any given $R$-module $M$. For the sake of concreteness we limit our attention to submodules of $R^n$, but the same proof applies to any finitely generated $R$-module $M$.

\begin{theorem} \label{evaluateZp}
Let $R\cong\Z_p^m$, let $M\in\mM(R^n)$. 
The number of $R$-submodules of $M$ of length $\ell$ is 
$$\begin{bmatrix} M \\ \ell \end{bmatrix}=\sum_{{\tiny \begin{array}{l} (\ell_1,\ldots,\ell_m)\in\N^m, \\
\ell_1+\ldots+\ell_m=\ell\end{array}}}\prod_{i=1}^m \qbin{\dim(e_iM)}{\ell_i}{p}.$$
In particular, $$\begin{bmatrix} M \\ 1 \end{bmatrix}=\sum_{i=1}^m \frac{p^{\dim(e_iM)}-1}{p-1}.$$
\end{theorem}

\begin{proof}
By Lemma~\ref{eis}, for all $M\in\mM(R^n)$ one has 
$$M=e_1M\oplus\ldots\oplus e_mM.$$
Therefore $\lambda(M)=\sum_{i=1}^m\lambda(e_iM).$ Moreover, for all $r=r_1e_1+\ldots+r_me_m\in R$ and $v\in M$ we have $re_iv=r_ie_iv$ for all $i$. Therefore the $R$-submodules of $e_iM$ coincide with its $\Z_p$-subspaces, hence $\lambda(e_iM)=\dim_{\Z_p}(e_iM)$. Hence we have shown that for every $R$-module $M\subseteq R^n$ $$\lambda(M)= \sum_{i=1}^m \dim_{\Z_p}(e_iM)=\dim_{\Z_p}(M).$$
Since for every collection of submodules $N_i\subseteq e_iM$ the module $N=N_1\oplus\ldots\oplus N_m\in\mM(M)$ and the $R$-submodules of $e_iM$ coincide with its $\Z_p$-subspaces, then the number of $R$-submodules of $M$ of length $\ell$ is
$$\begin{bmatrix} M \\ \ell \end{bmatrix}=\sum_{{\tiny \begin{array}{l}
(\ell_1,\ldots,\ell_m)\in\N^m, \\
\ell_1+\ldots+\ell_m=\ell\end{array}}}\prod_{i=1}^m \qbin{\dim(e_iM)}{\ell_i}{p}.$$
\end{proof}

Theorem~\ref{evaluateZp} allows us to evaluate the bounds of Theorem~\ref{sss} and Theorem~\ref{bound2} as follows.

\begin{corollary}\label{fourbounds}
Let $R\cong\Z_p^m$. Let $\mC\subseteq \mM(\Omega)$ be a submodule code with $\lambda(M)=k$ for all $M \in \mC$ and $d(\mC)=2\delta$. Let 
$$b(\lambda,k,\delta)=\min\left\{\sum_{{\tiny \begin{array}{l}
(\ell_1,\ldots,\ell_m)\in\N^m, \\
\ell_1+\ldots+\ell_m=k-\delta+1\end{array}}}\prod_{i=1}^m \qbin{u_i}{\ell_i}{p} : \ (u_1,\ldots,u_m)\in\N^m, u_1+\ldots+u_m=\lambda-\delta+1\right\}.$$
Then 
\begin{equation} \label{bb1}
|\mC| \le b(\lambda(\Omega),k,\delta).
\end{equation}
Moreover, 
\begin{equation}\label{bb2}
|\mC| \le \frac{\mathlarger{\sum}_{{\tiny \begin{array}{l}
(\ell_1,\ldots,\ell_m)\in\N^m, \\
\ell_1+\ldots+\ell_m=k-\delta+1\end{array}}}\mathlarger{\prod_{i=1}^m} 
\qbin{\dim(e_i\Omega)}{\ell_i}{p}}{b(k+\delta-1,k,\delta)}.
\end{equation}
If $\Omega=R^n$ and $\delta=k$, then
\begin{equation}\label{bb3}
|\mC| \le \frac{(p^n-1)/(p-1)}{\lceil (p^{k/m}-1)/(p-1)\rceil}.
\end{equation}
If in addition $m=2$ and $k$ is odd, then
\begin{equation}\label{bb4}
|\mC| \le \lfloor 2(p^{n}-1)/(p^h+p^{h-1}-2) \rfloor,
\end{equation}
where $h= \lceil k/2\rceil$.
\end{corollary}

\begin{proof}
Let $M \in \mM(\Omega)$ be a submodule of length $\lambda(M)=\lambda(\Omega)-\delta+1$ and let $u_i= \dim(e_iM)$ for all $i$. 
By Theorem \ref{evaluateZp}, the number of submodules of $M$ of length $k-\delta+1$ equals 
$$\begin{bmatrix} M \\ k-\delta+1 \end{bmatrix}=\sum_{{\tiny \begin{array}{l} (\ell_1,\ldots,\ell_m)\in\N^m, \\
\ell_1+\ldots+\ell_m=k-\delta+1\end{array}}}\prod_{i=1}^m \qbin{u_i}{\ell_i}{p}.$$
Hence Theorem \ref{sss} implies bound (\ref{bb1}). Similarly, bound (\ref{bb2}) follows from Theorem \ref{bound2}.

Now assume $\delta=k$ and $\Omega=R^n$. We have 
$$b(2k-1,k,k)= \min\left\{ \frac{1}{p-1} \left(
\sum_{i=1}^m p^{u_i} -m \right) : 
(u_1,...,u_m) \in \N^m,  \ \sum_{i=1}^m u_i =k \right\}.$$
Let $f: \R^m \to \R$ be the function defined by 
$f(x_1,...,x_m)= \sum_{i=1}^m p^{x_i}$ for all
$(x_1,...,x_m) \in \R^m$. Using e.g. the method of 
Lagrange multipliers from Calculus, one can show that the minimum of $f$ in the region
of $\R^m$ defined by the constraints
$$\sum_{i=1}^m x_i =k, \ \ \ \ \ \ x_i \ge 0 \mbox{ for all $i \in \{1,...,m\}$}
$$
is attained for $x_1=x_2= \ldots = x_m =k/m$, and that its value is
$mp^{k/m}$.
This shows that 
$$b(2k-1,k,k)
\ge  \lceil (mp^{k/m}-m)/(p-1) \rceil.$$
Bound (\ref{bb3}) now follows from bound (\ref{bb2}) and the fact that 
$\dim(e_i\Omega)=n$  for all $i \in \{1,...,m\}$.
If in addition $m=2$ and $k$ is odd, then without loss of generality we may assume $u_1 \ge u_2+1$.
Using elementary methods from Calculus, one shows that 
$$b(2k-1,k,k) \ge \frac{p^h + p^{h-1}-2}{p-1},$$
where $h= \lceil k/2\rceil$. This concludes the proof.
\end{proof}

We conclude this section by evaluating the bound of Theorem \ref{sss} for finite chain rings. 
We concentrate on codes $\mC \subseteq \mM(\Omega)$ with $\lambda(M)=k$ for all $M\in\mC$ and $d(\mC)=2k$.

\begin{theorem} \label{fcrlen1}
Let $R$ be a finite chain ring of length $e$, let $\pi \in R$ be a generator of the maximal ideal of $R$, let $q$ be the cardinality of the residue field of $R$. 
Let $$\Omega=R\times(\pi^{a_2})\times\ldots\times(\pi^{a_n})$$ for some $0\leq a_2\leq\ldots\leq a_n\leq e-1.$
Let $\mC \subseteq \mM(\Omega)$ be a submodule code with $d(\mC)=2k$ and whose codewords have length $k$. Then $$|\mC| \le \frac{q^{m}-1}{q-1}$$ where $m=\min\{1\leq i\leq n\mid (n-i)e-a_{i+1}-\ldots-a_n\leq k-1\}.$
In particular, if $\Omega=R^n$, then $$|\mC| \le \frac{q^{n-h+1}-1}{q-1}$$
where $k=he-r$ with $0\leq r\leq e-1$.
\end{theorem}

\begin{proof}
We claim that the submodules of length one of an $R$-module $M\subseteq R^n$ are in bijection with the vectors $v \in M$ of the form 
\begin{equation}\label{vectrref}
v= (\underbrace{0, \ldots, 0}_{i-1}, \pi^{e-1}, v_{i+1},...,v_{n})
\end{equation} 
where $1 \le i \le n$ and $v_{i+1},...,v_{n} \in (\pi^{e-1})$. In fact, every module of length one is minimally generated by one vector. If we represent a module generated by one vector by a matrix in reduced row-echelon form, then such a matrix has a unique row by Theorem~\ref{princ}, and by Proposition~\ref{criterio} the row is of the form $v=(0, \ldots, 0, \pi^s, v_{i+1},...,v_{n})$ with $v_{i+1},...,v_{n} \in (\pi^{s})$. Finally, $\lambda(\langle v\rangle)=\lambda(\pi^s)=e-s$ by Theorem~\ref{princ}. Hence the modules of length one are exactly those generated by vectors of the form (\ref{vectrref}). By uniqueness of the reduced row-echelon form, two such modules are distinct if and only if they are generated by distinct vectors in reduced row-echelon form. This proves the claim. Notice that there are exactly $q^{n-i}$ vectors of the form (\ref{vectrref}), for a fixed $i\in\{1,\ldots,n\}$.

Let $M\subseteq\Omega$ be a submodule of length $\lambda(\Omega)-k+1$. Let $m$ be the least integer such that $$M\subseteq R^{m}\times\underbrace{0\times\ldots\times 0}_{n-m}.$$ Notice that $m$ depends on $M$, and not just on its length. Then $M$ contains exactly $q^{m-1}+q^{m-2}+\ldots+q+1=\frac{q^m-1}{q-1}$ vectors of the form (\ref{vectrref}), since each such vector has $v_{i+1},\ldots,v_m\in (\pi^{e-1})$ and $v_{m+1}=\ldots=v_n=0$, for some $1\leq i\leq m$. Since $\Omega=R\times(\pi^{a_2})\times\ldots\times(\pi^{a_n})$, then $\lambda(\Omega)=ne-a_2-\ldots-a_n$. If $$M\subseteq (R^i\times 0\times\ldots\times 0)\cap \Omega=R\times\pi^{a_2}R\times\ldots\times\pi^{a_i}R\times 0\times\ldots\times 0$$ for some $i$, then $\lambda(M)=ne-a_2-\ldots-a_n-k+1\leq \lambda(R\times\pi^{a_2}R\times\ldots\times\pi^{a_i}R)=ie-a_2-\ldots-a_i$. Hence $$m\geq \min\{1\leq i\leq n : \ (n-i)e-a_{i+1}-\ldots-a_n\leq k-1\}$$ and equality holds for any $M\subseteq R\times(\pi^{a_2})\times\ldots\times(\pi^{a_m})\times 0\times\ldots\times 0$. By Theorem~\ref{sss} $$|\mC| \le \min \left\{\begin{bmatrix} M \\ k-\delta+1 \end{bmatrix}  \ : \ M \subseteq\Omega, \ \lambda(M)=\lambda(\Omega) -\delta +1\right\}$$ 
$$=\min \left\{\frac{q^i-1}{q-1}\ : \ M \subseteq\Omega\cap R^i\times 0\times\ldots\times 0, \ \lambda(M)=\lambda(\Omega) -\delta +1\right\}=\frac{q^m-1}{q-1}$$ where $m=\min\{1\leq i\leq n\mid (n-i)e-a_{i+1}-\ldots-a_n\leq k-1\}.$ 

If $\Omega=R^n$, then $a_2=\ldots=a_n=0$. Write $k=he-r$ with $0\leq r\leq e-1$. Then $(h-1)e\leq k-1\leq he-1$, hence $m=\min\{1\leq i\leq n\mid (n-i)e\leq k-1\}=n-h+1.$
\end{proof}

\section{Constructions} \label{seccostr}

In this section we propose some constructions for submodule codes for an ambient space of the form $\Omega=R^n$. 
Throughout the section we say that a code is asymptotically optimal if its cardinality asymptotically meets one of the bounds of the previous section.
We say that a code is optimal if its cardinality exactly meets one of the bounds.

We first concentrate on finite chain rings, and show how to construct optimal codes of maximum correction capability. Our codes can be regarded as the submodule code analogue of the partial spread codes from~\cite{ps}. We then look at finite PIR's that contain a field $\F$, and show how subspace codes over $\F$ can be lifted to submodule codes over $R$ by tensoring them with $R$. This allows us to construct optimal submodule codes over $\Z_p[i]$ of maximum correction capability. Finally we show how to obtain submodule codes over a ring of the form $R_1 \times\ldots\times R_m$ from submodule codes over $R_1,...,R_m$. We propose two constructions of the latter type, and discuss their decoding. For the first construction we take a cartesian product of codes on the $R_i$'s and show that this yields a code on $R$, whose parameters are determined and whose decoding can be reduced to decoding on each of the $R_i$'s. However, not every code over a $R$ is a product of codes over the $R_i$'s: Our second construction yields a code over $R$ which is not a product. We show that, in that case, decoding cannot be reduced to decoding on each of the $R_i$'s. This shows in particular that, although every finite PIR $R$ is isomorphic to a product of finite fields and finite chain rings, the study of codes over $R$ cannot be reduced in general to the study of codes over finite fields and finite chain rings.

\subsection{Partial spreads over finite chain rings}

We start with a construction that can be applied to any PIR. Its optimality relies on the existence of large sets of matrices, in which the length of the difference of any two of them is maximum. In Proposition \ref{lenmetric1} we will show that such large sets can be constructed over any finite chain ring.

\begin{theorem} \label{costruzione}
Let $R$ be a finite PIR. Let $n,k$ be positive integers. Write $k=h \cdot \lambda(R)-r$ with $0 \le r \le \lambda(R)-1$, and assume $n \ge 2h$.  
Write $n=\nu \cdot h+\rho$ with $0 \le \rho \le h-1$. Let $\mA \subseteq R^{h \times h}$ and $\mA' \subseteq R^{h \times (h+\rho)}$ be subsets such that 
$\lambda(\row(A-B))=h \cdot \lambda(R)$ for all $A,B \in \mA$ with $A \neq B$ and $A,B \in \mA'$ with $A \neq B$. 
For $i \in \{1,...,\nu-1\}$ let
$$\mS_i= \left\{ \begin{bmatrix}
              0_{h \times h} & \ldots & 0_{h \times h} & I_{h \times h} &
              A_{i+1} & \ldots & A_{\nu-1} & A_\nu
             \end{bmatrix} \ : \ A_{i+1},...,A_{\nu-1} \in \mA, \ A_\nu \in \mA'
\right\} \subseteq R^{h \times n},$$
where $I_{h \times h}$ is the identity matrix of size $h \times h$. Define 
$\mS_\nu = \left\{ \begin{bmatrix}
              0_{h \times (n-h)} & I_{h \times h}
             \end{bmatrix} \right\}$.
Let $\zeta \in R$ generate an ideal of length $\lambda(R)-r$. 
For all $i \in \{1,...,\nu\}$ let
$$\mS_{i,\zeta} = \{M_{\zeta} : M \in \mS_i\},$$
where $M_{\zeta}$ is the matrix obtained from $M$ by multiplying its last row by $\zeta$.
Then 
$$\mC = \bigcup_{i=1}^{\nu} \{\mbox{row}(M) : M \in \mS_{i,\zeta}\}$$
is a submodule code of length $\lambda(\mC)=k$, minumum distance $d(\mC)=2k$, and  cardinality
$$|\mC| = |\mA'| \cdot \frac{|\mA|^{\nu-1} -1}{|\mA| -1} + 1.$$
\end{theorem}
 
\begin{proof}
Let $M_{\zeta}\in\mS_{i,\zeta}$. Then $M_{\zeta}$ is in row-echelon form and $\lambda(\row(M_\zeta))= (h-1) \cdot \lambda(R) + \lambda(\zeta) = k$ by Theorem \ref{princ}. 
Define the code $$\mC' = \bigcup_{i=1}^{\nu} \{\mbox{row}(M) : M \in \mS_i\}.$$
Again $M$ is in row-echelon form and $\lambda(\row(M))= h \cdot \lambda(R)$. Moreover, arguing as in \cite[Theorem 13 and Proposition 17]{ps} and replacing the rank with the length, one sees that $d(\mC')=2h \cdot \lambda(R)$, i.e., the submodules that constitute $\mC'$ have trivial pairwise intersections.
Moreover, $$|\mC'| = |\mA'| \cdot \frac{|\mA|^{\nu-1} -1}{|\mA| -1} + 1.$$
Now observe that $\mC$ is obtained from $\mC'$ by taking an appropriate submodule of each codeword. Therefore the codewords of $\mC$ have trivial pairwise intersection. Hence $d(\mC)=2k$ and $|\mC|=|\mC'|$.
\end{proof}

We now show that over a finite chain ring $R$ one can construct  large sets $\mA$ and $\mA'$ to be used within the construction from Theorem~\ref{costruzione}.
  
\begin{lemma} \label{atmost}
Let $R$ be a ring, let $s,t >0$. Then for any $v_1,...,v_t \in R^s$ we have $\lambda(\langle v_1,...,v_t \rangle) \le t \cdot \lambda(R)$. 
\end{lemma}

\begin{proof}
Let $A\in R^{t\times s}$ be the matrix with rows $v_1,...,v_t$.
Right multiplication by $A$ induces an $R$-homomorphism $f_A: R^t \to R^s$, whose image is $\langle v_1,\ldots,v_t \rangle$. Since $\im(f_A)\cong R^t/\ker(f_A)$, then $\lambda(\langle v_1,\ldots,v_t \rangle)=\lambda(R^t)- \lambda(\ker(f_A)) \le \lambda(R^t) = t \cdot \lambda(R)$.   
\end{proof}

\begin{proposition} \label{lenmetric1}
Let $R$ be a finite chain ring with residue field of order  $q$. Then for all $h>0$ and for all $0 \le \rho \le h-1$ there exists a set $\mA \subseteq R^{h \times (h+\rho)}$ with 
$|\mA|=q^{h+\rho}$ and $\lambda(\row(A-B))= h \cdot \lambda(R)$ for all $A,B \in \mA$ with $A \neq B$.
\end{proposition}

\begin{proof}
We first prove the statement for $\rho=0$. Let $\pi$ be a generator of the maximal ideal of $R$, and let $f: R \to R/(\pi)$ be the projection map. Let $\iota: R/(\pi) \to R$ be such that $f \circ \iota$ is the identity of $R/(\pi)$. Such a $\iota$ can be obtained by mapping each element of $R/(\pi)$ to one of its representatives in $R$. Notice that we do not require that $\iota$ is a ring homomorphism. 
We extend $f$ and $\iota$ entrywise to ${f}:R^{h \times h} \to (R/(\pi))^{h \times h}$ and ${\iota} : (R/(\pi))^{h \times h} \to R^{h \times h}$.
  
Since $(R/(\pi))^{h \times h} \cong \F_q^{h \times h}$, by \cite[Section 6]{del1} there exists $\mA' \subseteq (R/(\pi))^{h \times h}$ with $|\mA'|=q^h$ and $A'-B'$ invertible for any $A',B' \in \mA'$ with $A' \neq B'$. Then the set of matrices $\mA=\{{\iota}(A') : A' \in \mA'\} \subseteq R^{h \times h}$ has the expected properties. Indeed, let $A', B' \in \mA'$ with $A' \neq B'$. Since $f$ is a ring homomorphism, we have  
$$f(\det({\iota}(A')-{\iota}(B')))=\det({f}({\iota}(A'))-{f}({\iota}(B')))=\det(A'-B') \neq 0.$$
Therefore $\det({\iota}(A')-{\iota}(B')) \notin (\pi)$. 
As $(\pi)$ is the only maximal ideal of $R$, $\det({\iota}(A')-{\iota}(B'))$ is invertible, hence ${\iota}(A')-{\iota}(B')$ is invertible. 
This implies $\row({\iota}(A')-{\iota}(B'))\cong R^h$, which has length $h \cdot \lambda(R)$.
In addition $|\mA|=|\mA'|=q^h$.
   
Now assume $\rho>0$, and set $h'=h+\rho$. By the first part of the proof there exists a set of matrices $\mB \subseteq R^{h' \times h'}$ with $\lambda(\row(A-B))=h' \cdot \lambda(R)$ for all $A,B \in \mB$ with $A \neq B$. For $A \in \mB$ denote by $\overline{A}$ the matrix obtained from $A$ by deleting the first $\rho$ rows. A simple application of Lemma \ref{atmost} shows that the set $\mA=\{\overline{A} : A \in \mB\} \subseteq R^{h \times h'}$ has the desired properties.
\end{proof}

\begin{example} \label{excostrFCR}
Let $R$ be a finite chain ring with residue field of order $q$. Following the notation of Theorem \ref{costruzione}, Proposition \ref{lenmetric1} allows us to construct a submodule code $\mC \subseteq \mM(R^n)$ of constant length $\lambda(\mC)=k$, minimum distance $d(\mC)=2k$, and cardinality
$$|\mC| = q^{h+\rho} \cdot \frac{q^{h(\nu-1)}-1}{q^h-1} +1 = \frac{q^n-q^{h+\rho}+q^h-1}{q^h-1}\in \mathcal{O}(q^{n-h})$$
(as $n>h+\rho$). 
Let $\overline{\mC}$ be a submodule code with the same parameters as $\mC$ and maximum cardinality. By Theorem \ref{fcrlen1} we have $|\overline{\mC}| \le (q^{n-h+1}-1)/(q-1)$. Therefore
$$1 \ge \frac{|\mC|}{|\overline{\mC}|} \ge \frac{q^n-q^{h+\rho}+q^h -1}{q^h-1} \cdot\frac{q-1}{q^{n-h+1}-1} \stackrel{q \to \infty}{\longrightarrow} 1.$$
This shows that $\mC$ is an asymptotically optimal submodule code.
\end{example}

\subsection{Tensor product construction and partial spreads over rings of the form $\Z_p^m$}

Assume that $R$ contains a finite field $\F \subseteq R$ as a subring and that $R$ and $\F$ have the same one. Let $V \subseteq \F^n$ be an $\F$-linear space. 
Recall that the tensor product $V \otimes_{\F}R \subseteq R^n$ is the submodule of $R^n$ generated by the elements of $V$. If $V=\langle v_1,\ldots,v_m\rangle_\F$, then 
$$V \otimes_{\F}R=\langle v : \ v\in V\rangle_R=\langle v_1,\ldots,v_m \rangle_R.$$

\begin{lemma} \label{lemmC1}
Let $V \subseteq \F^n$ be an $\F$-linear space. Then 
$$\lambda(V \otimes_{\F}R) = \lambda(R) \cdot \dim_{\F}(V).$$
\end{lemma}

\begin{proof}
Let $t=\dim_{\F}(V)$, and let $A \in \F^{t \times n}$ be a matrix in reduced row-echelon form, whose rows generate $V$. 
Regard $A$ as a matrix over $R$. Then $A$ is still in row-echelon form and $\row(A)=V \otimes_{\F}R$. 
Since all the pivots of $A$ are ones, by Theorem \ref{princ} we have $\lambda(V \otimes_{\F}R) = \lambda(1) \cdot t = \lambda(R)\cdot \dim_{\F}(V)$, as claimed.  
\end{proof}

\begin{lemma}\label{lemmC2}
Let $V,W \subseteq \F^n$ be $\F$-linear spaces. Then 
$(V \otimes_{\F}R) \cap (W \otimes_{\F}R) = (V \cap W) \otimes_{\F}R$.
\end{lemma}

\begin{proof}
By definition $(V \cap W) \otimes_{\F}R\subseteq (V \otimes_{\F}R) \cap (W \otimes_{\F}R)$. Therefore by Lemma \ref{prellen} it suffices to show that they have the same length. By Lemma~\ref{lemmC1} 
$$\lambda((V \cap W) \otimes_{\F}R)=\lambda(R)\cdot\dim_\F(V\cap W)=\lambda(R)\cdot(\dim_\F(V)+\dim_\F(W)-\dim_\F(V+W))=$$
$$\lambda(V\otimes_{\F}R)+\lambda(W\otimes_{\F}R)-\lambda((V+W)\otimes_{\F}R)=\lambda((V \otimes_{\F}R) \cap (W \otimes_{\F}R)),$$
where the last equality follows from Lemma~\ref{prellen} and from observing that $(V+W)\otimes_\F R=V\otimes_\F R + W\otimes_\F R.$
\end{proof}

From Lemma \ref{lemmC1} and \ref{lemmC2} one obtains the following construction.

\begin{theorem} \label{costrtensore}
Let $\mC \subseteq \mM(\F^n)$ be a subspace code of minimum subspace distance $2\delta$ and $\dim_{\F}(V)=k$ for all $V \in \mC$. Then 
$$\mC \otimes_{\F}R = \{V \otimes_{\F}R : V \in \mC \} \subseteq \mM(R^n)$$
is a submodule code of cardinality $|\mC \otimes_{\F}R|=|\mC|$, whose codewords have length $\lambda(R) \cdot k$, 
and whose minimum distance is $2 \lambda(R) \cdot\delta$.
\end{theorem}

\begin{example}
Let $\F=\Z_p$ and $R=\Z_p^m$. Then $\Z_p \cong \{(a,a,...,a) : a \in \Z_p\}\subseteq R$ can be viewed as a subring of $R$. We have $\lambda(R)=m$. Let $h$ be an integer, $1\leq h\leq n/2$. By \cite[Theorem 13 and Proposition 17]{ps}, there exists a subspace code $\mC \subseteq \mM(\Z_p^n)$ of constant dimension $h$, minimum distance $2h$ and cardinality $(p^n-p^{h+\rho} + p^h -1)/(p^h-1)$, where $\rho$ is the remainder of the division of $n$ by $h$. By Theorem \ref{costrtensore}, $\mC \otimes_{\Z_p}R \subseteq \mM(R^n)$ is a submodule code whose codewords have length $mh$ and minimum distance $2mh$. Moreover, $$|\mC \otimes_{\Z_p}R| = \frac{p^n-p^\rho}{p^h-1}-p^\rho+1.$$

Let $\overline{\mC}$ be a submodule code with the same parameters as $\mC\otimes_{\Z_p} R$ and maximum cardinality. By Theorem \ref{fourbounds} (\ref{bb3}), $|\overline{\mC}| \le (p^n-1)/(p^h-1)$. Therefore
$$1 \ge \frac{|\mC\otimes_{\Z_p} R|}{|\overline{\mC}|} \ge \frac{p^n-p^{h+\rho}+p^h -1}{p^h-1} \cdot \frac{p^h-1}{p^n-1} \stackrel{q \to \infty}{\longrightarrow} 1.$$
Hence $\mC\otimes_{\Z_p} R$ is asymptotically optimal, and it is optimal when $\rho=0$.

Fix $1\leq \ell\leq m-1$. For all $M \in \mC\otimes_{\Z_p} R$ choose a submodule $M' \subseteq M$ with $\lambda(M')=mh-\ell$. Then 
$\mD = \{M' : M \in \mC\otimes_{\Z_p} R\} \subseteq  \mM(R^n)$ is a submodule code with minimum distance $2mh-2\ell$ and whose codewords have length $mh-\ell$. Moreover,
$$|\mD|=|\mC\otimes_{\Z_p} R| = \frac{p^n-p^\rho}{p^h-1}-p^\rho+1.$$ Let $\overline{\mD}$ be a submodule code with the same parameters as $\mD$ and maximum cardinality. By Theorem~\ref{fourbounds} (\ref{bb3}) we have $|\overline{\mD}| \le (p^n-1)/(p^{h-\ell/m}-1)$. Therefore
$$1 \ge \frac{|\mD|}{|\overline{\mD}|} \ge \frac{p^n-p^{h+\rho}+p^h -1}{p^h-1} \cdot \frac{p^{h-\ell/m}-1}{p^n-1} \in \mathcal{O}\left(p^{-\ell/m}\right).$$
If in addition $m=2$, then by Proposition~\ref{fourbounds} (\ref{bb4})
$$1 \ge \frac{|\mD|}{|\overline{\mD}|} \ge \frac{p^n-p^{h+\rho}+p^h -1}{p^h-1} \cdot\frac{p^h+p^{h-1}-2}{2(p^n-1)} \stackrel{q \to \infty}{\longrightarrow} 1/2.$$
Therefore $\mD$ is asymptotically optimal, up to a factor 2.   
\end{example}

\begin{remark}\label{morethanvs}
Let $\mC \subseteq \mM(\F^n)$ be a subspace code over the finite field $\F$, and let $W \subseteq \F^n$ be a decodable space for the code $\mC$, i.e., an $\F$-linear space for which there exists $V \in \mC$ with $d(V,W) \le \lfloor (d_{\textnormal{S}}(\mC)-1)/2 \rfloor$, where $d_{\textnormal{S}}(\mC)$ denotes the minimum subspace distance of $\mC$. Then $W \otimes_{\F}R$ is decodable in the submodule code $\mC \otimes_{\F}R$, and it decodes to $V\otimes_{\F}R$.

However, there exist submodules of $R^n$ which are decodable in $\mC\otimes_{\F}R$ but are not of the form $W\otimes_{\F}R$, with $W$ an $\F$-space which is decodable in $\mC$. Moreover, if $N \subseteq R^n$ is decodable in $\mC\otimes_{\F}R$, then $N \cap \F^n$ is not necessarily decodable in $\mC$.

Let e.g. $R=\Z_5[i] $, $\F=\Z_5$, and let $\mC=\{V_1,V_2\} \subseteq \mM(\Z_5^4)$ be the subspace code 
whose codewords are the 2-dimensional spaces
$$V_1= \langle (1,0,1,0), (0,1,0,1)\rangle_{\Z_5}, \ \ \ \ \ \ V_2= \langle (1,0,2,1), (0,1,1,0)\rangle_{\Z_5}.$$
Then $\mC$ has subspace distance $d_{\textnormal{S}}(\mC)=4$.  
By Theorem \ref{costrtensore} the submodule code $\mC \otimes_{\Z_5} \Z_5[i]$ has two codewords of length
$4$ and submodule distance $d(\mC \otimes_{\Z_5} \Z_5[i])=8$.

Let $$N= \langle (i+2,0,i+2,0), (0,1,0,i-1)\rangle_{\Z_5[i]} \subseteq \Z_5[i]^4$$ be a received submodule.
Then $d(N,V_1 \otimes_{\Z_5} \Z_5[i]) = 3 \le \lfloor (8-1)/2 \rfloor =3$, so $N$ is decodable in $\mC \otimes_{\Z_5} \Z_5[i]$. However, $N$ is not of the form 
$W \otimes_{\Z_5} \Z_5[i]$ for any vector space $W$. Moreover, $N \cap \Z_5^4 = 0$, so $N \cap \Z_5^4$ is not decodable in $\mC$ with respect to the subspace distance.
\end{remark}

\subsection{Two constructions over products of rings}

Let $R\cong R_1 \times \ldots \times R_m$, where $R_1,\ldots,R_m$ are finite commutative rings with identity. Let $\pi_i$ be the projection on the factor $R_i$. Then each $\pi_i$ extends componentwise to a map $\pi_i:R^n \to R_i^n$. 

\begin{lemma} \label{decomp}
Let $R\cong R_1 \times \ldots \times R_m$ be a finite ring, let $M \subseteq R^n$ be an $R$-module. 
Then $M\cong \pi_1(M)\times\ldots\times\pi_m(M)$, each $\pi_i(M)$ is an $R_i$-module and $$\lambda_R(M) = \sum_{i=1}^m \lambda_{R_i}(\pi_i(M)).$$
\end{lemma}

\begin{proof}
Let $r_i\in R_i$ and $v \in M$. Then $\pi_i(v)=(\pi_i(v_1),\ldots,\pi_i(v_n))\in \pi_i(M)\subseteq R_i^n$, and $r_i\pi_i(v)=(r_i\pi_i(v_1),\ldots,r_i\pi_i(v_n))\in \pi_i(M)$.
This makes $\pi_i(M)$ into an $R_i$-module. Moreover, the isomorphism $R^n\cong R_1^n\times\ldots\times R_m^n$ restricts to an isomorphism $M\cong\pi_1(M)\times\ldots\times\pi_m(M)$. Hence 
$$\lambda_R(M)= \sum_{i=1}^m \lambda_R(\pi_i(M))= \sum_{i=1}^m\lambda_{R_i}(\pi_i(M)),$$
where the last equality follows from the fact that any $R$-submodule of $\pi_i(M)$ is an $R_i$-submodule, and viceversa.
\end{proof}


We start with a simple construction, where we produce a code over $R_1\times\ldots\times R_m$ by taking the cartesian product of codes over each $R_i$. For simplicity of notation we identify $R^n$ and $R_1^n\times\ldots\ R_m^n$.

\begin{theorem} \label{prodcostr}
Let $R_1,...,R_m$ be finite PIR's and let $R=R_1 \times\ldots\times R_m$. For $i \in \{1,...,m\}$ let $\mC_i \subseteq \mM(R_i^n)$ be a submodule code whose codewords have length $k_i$. Then 
$$\mC=\mC_1 \times \ldots \times \mC_m = \{M_1 \times \ldots \times M_m \ :\  M_i\in \mC_i \mbox{ for all $i$ }\} \subseteq \mM(R^n)$$
is a submodule code of cardinality $|\mC|= |\mC_1|  \ldots |\mC_m|$, whose codewords have length $k_1+\ldots+k_m$, and with minimum distance $d(\mC)=\min\{d(\mC_i) : \ 1 \le i \le m\}$.
\end{theorem}

We now show that decoding of the {\bf product code} $\mC=\mC_1 \times \ldots \times \mC_m$ over $R$ can be reduced to decoding each of the codes $\mC_i$ over $R_i$.

\begin{proposition} \label{decodingproduct}
Let $R_1,...,R_m$ be finite PIR's and let $R=R_1 \times\ldots\times R_m$. For $i \in \{1,...,m\}$ let $\mC_i \subseteq \mM(R_i^n)$ be a submodule code and let $\mC=\mC_1 \times \ldots \times \mC_m\subseteq\mM(R^n)$ be the product code. Let $N \subseteq R^n$ be a received decodable submodule, i.e., an $R$-module for which there exists an $M=M_1 \times \ldots \times M_m \in\mC$ such that $d(N,M) \le  \lfloor (d(\mC)-1)/2 \rfloor$.
Then for all $i \in \{1,...,m\}$ we have $d(\pi_i(N),M_i) \le \lfloor (d(\mC_i)-1)/2\rfloor$, i.e., $\pi_i(N)$ decodes to $M_i$ in $\mC_i$.
\end{proposition}

\begin{proof}
By Lemma~\ref{decomp}, $N=\pi_1(N) \times \ldots \times \pi_m(N)$, and $M_i=\pi_i(M)$ for $1\leq i\leq m$. Moreover
\begin{eqnarray*}
d(N,M) &=& \lambda(N) + \lambda(M) - 2 \lambda(N \cap M) \\
&=& \sum_{i=1}^m \lambda_{R_i}(\pi_i(N)) + \sum_{i=1}^m \lambda_{R_i}(\pi_i(M)) -2 \sum_{i=1}^m \lambda_{R_i}(\pi_i(N \cap M))  \\
&\ge& \sum_{i=1}^m d(\pi_i(N),\pi_i(M)), 
\end{eqnarray*}
where the inequality follows from the fact that $\pi_i(N \cap M) \subseteq \pi_i(N) \cap \pi_i(M)$.
Therefore for all $i$ we have 
$$d(\pi_i(N),\pi_i(M)) \le d(N,M) \le \lfloor (d(\mC)-1)/2 \rfloor 
\le \lfloor (d(\mC_i)-1)/2 \rfloor,$$
hence $\pi_i(N)$ decodes to $M_i$ in $\mC_i$.
\end{proof}

Finally, we provide another construction which combines submodule codes over the factors $R_i$ into a submodule code over $R=R_1 \times \ldots \times R_m$. Compared to the product construction of Theorem~\ref{prodcostr}, this construction produces a code with smaller cardinality and larger minimum distance, whose decoding cannot be reduced to decoding over the $R_i$'s. Again, for simplicity we identify $R^n$ and $R_1^n\times\ldots\times R_m^n$.

\begin{theorem} \label{parallcostr}
Let $R_1,...,R_m$ be finite PIR's, and let $R=R_1 \times \ldots \times R_m$. For $i \in \{1,...,m\}$ let $\mC_i \subseteq \mM(R_i^n)$ be a submodule code whose codewords have length $k_i$. Let $c= \min |\mC_i|$, and for all $i \in \{1,...,m\}$ fix a subcode $\mC_i' \subseteq \mC_i$ with $|\mC_i'|=c$. Enumerate the elements of 
each $\mC_i'$ as $\mC_i'=\{M_{1,i},...,M_{c,i}\}$. Then 
$$\mC= \{M_{j,1} \times \ldots \times M_{j,m} \ :\  1 \le j \le c \} \subseteq \mM(R^n)$$
is a submodule code of cardinality $|\mC|=c$, with $d(\mC) \ge d(\mC_1)+\ldots+d(\mC_m)$, and whose codewords have length $k_1+\ldots+k_m$.
\end{theorem}

\begin{proof}
We only prove the part about the minimum distance. Let 
$j,j' \in \{1,...,c\}$ with $j \neq j'$. Arguing as in the proof of
Proposition \ref{decodingproduct}, one finds 
$$d(M_{j,1} \times \ldots \times M_{j,m} , \ M_{j',1} \times \ldots \times M_{j',m})
\ge \sum_{i=1}^m d(M_{j,i},M_{j',i}) \ge \sum_{i=1}^m d(\mC_i),$$
where the last inequality follows from the fact that $M_{j,i} \neq M_{j',i}$
whenever $j \neq j'$. 
\end{proof}

\begin{remark}
Notice that an $R$-module which is decodable with respect to the code $\mC$ constructed in Theorem~\ref{parallcostr} is not necessarily a product of $R_i$-modules that are decodable with respect to the codes $\mC_i$. 
E.g., let $m=2$, $n=4$, $R_1=R_2=\Z_2$, $R=\Z_2 \times \Z_2$. Let
$$M_{1,1} = M_{1,2} = \langle (1,0,1,0), \ (0,1,0,1)\rangle_{\Z_2}, \ \ \ \ 
M_{2,1} = M_{2,2} = \langle (1,0,1,1), \ (0,1,1,0) \rangle_{\Z_2}$$ and 
$\mC_1=\mC_1'=\{M_{1,1}, M_{1,2} \}$, $\mC_2=\mC_2'=\{M_{2,1}, M_{2,2} \}$. Then 
$$\mC= \left\{ \mbox{row} \begin{bmatrix}
(1,1) & (0,0) & (1,1) & (0,0) \\ 
(0,0) & (1,1) & (0,0) & (1,1)\end{bmatrix} \ , \ 
\mbox{row} \begin{bmatrix}
(1,1) & (0,0) & (1,1) & (1,1) \\ 
(0,0) & (1,1) & (1,1) & (0,0) \end{bmatrix}
 \right\}.$$
 The code $\mC$ has minimum distance $d(\mC)=8$. Let 
 $$N= \mbox{row} \begin{bmatrix}
 (1,0) & (0,0) & (1,0) & (0,0) \\
 (0,0) & (1,1) & (0,1) & (1,1)
 \end{bmatrix}$$
 be a received submodule. Then $N$ decodes to $M_{1,1} \times M_{1,2}\in \mC$, as 
 $d(N,M_{1,1} \times M_{1,2}) = 3 \le \lfloor (8-1)/2 \rfloor$. However, 
 $\pi_2(N) = \langle (0,1,1,1)\rangle_{\Z_2}$ is not decodable in $\mC_2$. 
 In fact, $d(\pi_2(N), M_2^1)= d(\pi_2(N), M_2^2)= 3$.
\end{remark}

\end{document}